\documentclass{article}

\usepackage[preprint]{ProbNum26} 
\usepackage{amsmath}
\usepackage{amsfonts}

\usepackage[round]{natbib}

\usepackage[capitalise,nameinlink]{cleveref}

\usepackage{amsthm}
\newtheorem{lemma}{Lemma}[section]  

\AddToHook{cmd/appendix/before}{%
    \crefalias{section}{appendix}%
    \crefalias{subsection}{appendix}
}

\usepackage{enumitem}

\newlist{enumeratlem}{enumerate}{1}
\setlist[enumeratlem]{label=\arabic*., ref=\thelemma.\arabic*}
\crefname{enumeratlemi}{Lemma}{Lemmas}

\DeclareMathOperator*{\argmax}{arg\,max}
\DeclareMathOperator*{\argmin}{arg\,min}
\usepackage{bm}

\usepackage{array}
\usepackage{booktabs}

\newcolumntype{L}[1]{>{\raggedright\arraybackslash}p{#1}}

\newcolumntype{R}[1]{>{\raggedleft\arraybackslash$}p{#1}<{$}}

\probnumtitle{Modified Bryson-Frazier Smoothing and Hyperparameter Learning for Temporal Gaussian Process Regression}
\probnumauthors{%
\name{Tom Colemont}%
\affiliation{Leuven Gravity Institute, Department of Electrical Engineering (STADIUS),  KU Leuven, Belgium}
\and%
\name{Brecht Evens}%
\affiliation{Leuven Gravity Institute, Department of Electrical Engineering (STADIUS), KU Leuven, Belgium}
\and%
\name{Tjonnie G.-F. Li}%
\affiliation{Leuven Gravity Institute, Department of Physics \& Astronomy, KU Leuven, Belgium}
\and%
\name{Frederik De Ceuster}%
\affiliation{Leuven Gravity Institute, Department of Physics \& Astronomy, KU Leuven, Belgium}%
\\[2pt]
}

\probnumabstract{One-dimensional Gaussian processes with stationary, integrable kernel functions admit exact or arbitrarily accurate state-space representations, enabling linear-time inference through Kalman filtering and Rauch-Tung-Striebel (RTS) smoothing. However, the RTS smoother requires inversion of predicted state covariance matrices, which can become ill-conditioned and may therefore lead to numerical instabilities. In this work, we revisit the modified Bryson-Frazier (MBF) smoother as an alternative to the RTS smoother for Gaussian process regression in its state-space representation. In addition to reducing computational cost and memory requirements, the MBF smoother computes the same posterior distributions as the RTS smoother while avoiding the problematic covariance matrix inversion and the associated numerical instabilities. Furthermore, we demonstrate that the intermediate quantities computed by the MBF smoother can be reused to compute gradients of the negative log marginal likelihood, enabling kernel hyperparameter learning with minimal additional cost. Together, these results establish the MBF smoother as a unified and numerically robust approach to inference and kernel hyperparameter learning for one-dimensional Gaussian process regression.}

\begin{document}

\section{Introduction}
\label{sec:introduction}
\setlength{\parskip}{8pt}
Gaussian processes (GPs) provide flexible and nonparametric models for regression, and play a central role in probabilistic numerics (PN).
However, standard GP inference has a cubic complexity in the number of data points, limiting scalability \citep{Rasmussen_GP}.
A common strategy to overcome this limitation is to leverage the class of Markovian kernel functions.
For such kernels, GP regression can be reformulated as inference in an equivalent linear Gaussian state-space model, enabling linear-time inference through Bayesian filtering and smoothing \citep{Sarkka_SSM}.

Since these state-space models are linear and Gaussian, the filtering and smoothing distributions admit closed-form recursions: the Kalman filter and Rauch-Tung-Striebel (RTS) smoother, respectively \citep{kalman1960new, rauch1965maximum}. Although exact finite-dimensional state-space representations are available only for Markovian kernels, any one-dimensional GP with a stationary, integrable kernel can be approximated arbitrarily well by such models \citep{LEG_model}. The combination of linear-time inference and broad kernel applicability has made this approach widely used in PN, particularly in probabilistic ODE solvers \citep{Tronarp_Probabilistic_ODEs}.

Despite reducing the computational complexity in the number of data points from cubic to linear, Kalman filtering and RTS smoothing remain cubic in the state-space dimension.
This dependence on the state-space dimension becomes problematic in high-dimensional settings, such as spatio-temporal regression. In response, computation-aware Kalman filters and smoothers have been developed that approximate both the filtering and smoothing distributions while tracking the induced approximation errors \citep{pfortner2024computation}.

A second, distinct limitation concerns numerical stability. The Kalman filter and, in particular, the RTS smoother may suffer from numerical instabilities when state covariance matrices become ill-conditioned or effectively singular. Such regimes arise naturally when PN methods aim to compute highly accurate solutions, where the posterior uncertainty may become very small. This is especially problematic for RTS smoothing, which requires inversion of the predicted state covariance matrix. When this matrix is ill-conditioned or singular, the RTS recursions may become unstable and break down. Coordinate transformations and square-root formulations can mitigate these issues, but incur additional computational cost \citep{Kraemer_Stable}.

In this paper, we revisit a classical alternative to the RTS smoother from control and estimation theory: the modified Bryson-Frazier (MBF) smoother.
Originally proposed by \citet{Bierman1973}, this reformulation of the RTS smoother avoids the problematic inversion of the predicted covariance matrix, while additionally reducing computational and memory requirements. Closely related smoothing formulations have appeared under different names, including the \textit{disturbance smoother} of \citet{kohn1989fast} and, more recently, the \textit{inversion}\textit{-free RTS smoother} in \citet{pfortner2024computation} in the context of computation-aware methods, where its favorable numerical properties were also noted.

This paper makes the following contributions: \setlength{\parskip}{.5pc}
\begin{enumerate}
    \item We present the MBF smoother as a numerically stable alternative to the RTS smoother for state-space inference in PN, retaining linear complexity.
    \item We provide a unified dynamic-programming view of Kalman filtering, RTS, and MBF smoothing, showing that the MBF recursions arise as the adjoint equations of the same optimization problem.
    \item We show that the MBF variables can be reused to compute gradients of the negative log marginal likelihood, enabling efficient hyperparameter optimization with minimal additional cost. 
\end{enumerate}
Together, these results establish the MBF smoother as a unified and numerically robust approach to inference and hyperparameter optimization for one-dimensional GP regression. Although we focus on regression for clarity, all results naturally extend to probabilistic numerical methods based on state-space inference.

\section{Inference in State-Space Models}
\label{sec:inference_in_state_space_models}
State-space methods provide an efficient computational framework for performing inference in GP models with Markovian structure. This section summarizes the conventional filtering and smoothing algorithms, and discusses their numerical aspects.

As shown by \citet{Sarkka_SSM, LEG_model}, any one-dimensional GP with a stationary, integrable kernel can be represented either exactly or approximated arbitrarily well by a linear Gaussian state-space model,
\begin{align}
    \mathbf{x}_{k} &= \mathbf{A}_{k-1} \mathbf{x}_{k-1} + \mathbf{q}_{k-1} \\
    \mathbf{y}_k &= \mathbf{H}_k \mathbf{x}_k + \mathbf{r}_k
\end{align}
where $\mathbf{x}_k \in \mathbb{R}^d$ denotes the latent state at time step $k$, $\mathbf{y}_k \in \mathbb{R}^m$ the measurement, and $\mathbf{q}_k$ and $\mathbf{r}_k$ the process and measurement noise, respectively,
\begin{align}
    \mathbf{q}_k &\sim \mathcal{N}(\mathbf{0}, \mathbf{Q}_k) \\
    \mathbf{r}_k &\sim \mathcal{N}(\mathbf{0}, \mathbf{R}_k)
\end{align}
with covariance matrices $\mathbf{Q}_k \in \mathbb{R}^{d \times d}$ and $\mathbf{R}_k \in \mathbb{R}^{m \times m}$. Together with the transition matrix $\mathbf{A}_{k} \in \mathbb{R}^{d \times d}$ and measurement matrix $\mathbf{H}_k \in \mathbb{R}^{m \times d}$, these quantities define the state-space model. This equivalent representation reformulates GP regression as a state estimation problem, which can be solved efficiently through Kalman filtering and RTS smoothing \citep{Sarkka_BFS}.

\subsection{Kalman Filtering}
For linear Gaussian state-space models, Bayesian filtering admits a closed-form solution, known as the Kalman filter \citep{kalman1960new}. The filter starts from an initial distribution $\mathbf{x}_0 \sim \mathcal{N}(\mathbf{m}_0, \mathbf{P}_0)$, where $\mathbf{P}_0$ is the stationary prior covariance. Assuming that $\mathbf{x}_0$, $\mathbf{q}_k$ and $\mathbf{r}_k$ are mutually independent, it proceeds in two steps:

\textbf{Predict:} the predicted distribution of the latent state is recursively determined as:
\begin{align}
    \mathbf{x}_k \mid \mathbf{y}_{1:k-1} &\sim \mathcal{N}(\mathbf{m}_k^-, \mathbf{P}_k^-) \\
    \mathbf{m}_{k}^- &= \mathbf{A}_{k-1} \mathbf{m}_{k-1}^{}\\
    \mathbf{P}_{k}^- &= \mathbf{A}_{k-1} \mathbf{P}_{k-1} \mathbf{A}_{k-1}^\top + \mathbf{Q}_{k-1}
\end{align}

\textbf{Filter:} the filtered distribution of the latent state incorporates a new measurement and reads:
\begin{align}
    \mathbf{x}_k \mid \mathbf{y}_{1:k} &\sim \mathcal{N}(\mathbf{m}_k, \mathbf{P}_k) \\
    \mathbf{z}_k &= \mathbf{y}_k - \mathbf{H}_k \mathbf{m}_{k}^- \\
    \mathbf{S}_k &= \mathbf{H}_k \mathbf{P}_{k}^- \mathbf{H}_k^\top + \mathbf{R}_k \\
    \mathbf{K}_k &= \mathbf{P}_{k}^- \mathbf{H}_k^\top \mathbf{S}_k^{-1} \\
    \mathbf{m}_{k} &= \mathbf{m}_{k}^- + \mathbf{K}_k \mathbf{z}_k \\ \label{eq:cov_KF}
    \mathbf{P}_{k} &= (\mathbf{I} - \mathbf{K}_k \mathbf{H}_k) \mathbf{P}_{k}^-
\end{align}
If no measurement is available at time step $k$, the update step is omitted and the filtered distribution coincides with the predicted distribution, effectively extrapolating the state estimate based on prior information. In conclusion, the Kalman filter computes the distribution of the latent state at time step $k$, conditioned on all measurements up to and including that time step.

\subsection{Rauch-Tung-Striebel Smoothing}
Bayesian smoothing determines the smoothing distribution of the latent state at time step $k$ conditioned on all measurements up to and including a final time step~$T$. For linear Gaussian state-space models, it admits a closed-form solution, known as the RTS smoother \citep{rauch1965maximum}. Starting from the final filtered state $\mathbf{x}_T \sim \mathcal{N}(\mathbf{m}_T, \mathbf{P}_T)$, a backward pass from $k=T-1 \ldots 0$ uses the predicted and filtered distributions to compute the smoothing distributions as:
\begin{align}
    \mathbf{x}_k \mid \mathbf{y}_{1:T} &\sim \mathcal{N}(\mathbf{m}_{k}^s, \mathbf{P}_{k}^s) \\
    \mathbf{J}_k &= \mathbf{P}_{k} \mathbf{A}_k^\top [\mathbf{P}_{k+1}^-]^{-1} \\
    \mathbf{m}_{k}^s &= \mathbf{m}_{k} + \mathbf{J}_k (\mathbf{m}_{k+1}^s - \mathbf{m}_{k+1}^-) \\ \label{eq:cov_RTS}
    \mathbf{P}_{k}^s &= \mathbf{P}_{k} + \mathbf{J}_k (\mathbf{P}_{k+1}^s - \mathbf{P}_{k+1}^-) \mathbf{J}_k^\top
\end{align}

\subsection{Numerical Considerations}
Although Kalman filtering and RTS smoothing provide exact inference for linear Gaussian state-space models, their practical implementation may suffer from numerical instabilities.

A first source of numerical instability arises in the covariance update equations of both the Kalman filter and the RTS smoother. In their standard form, these updates involve repeated subtractions of positive definite matrices. Due to the accumulation of numerical errors, the resulting matrices may lose positive definiteness. In addition, the standard Kalman filter covariance update does not enforce symmetry by construction. To mitigate these issues, the covariance updates can be reformulated in their so-called Joseph form \citep{Kraemer_Stable}:
\begin{align}
    \mathbf{B}_k &= \mathbf{I} - \mathbf{K}_k \mathbf{H}_k \\
    \mathbf{G}_k &= \mathbf{I} - \mathbf{J}_k \mathbf{A}_k \\
    \mathbf{P}_{k} &={} \mathbf{B}_k \mathbf{P}_{k}^- \mathbf{B}_k^\top + \mathbf{K}_k \mathbf{R}_k \mathbf{K}_k^\top \\
\begin{split}
    \mathbf{P}_{k}^s &={} \mathbf{G}_k \mathbf{P}_{k} \mathbf{G}_k^\top +  \mathbf{J}_k \mathbf{Q}_k \mathbf{J}_k^\top + \mathbf{J}_k \mathbf{P}_{k+1}^s  \mathbf{J}_k^\top
\end{split}
\end{align}
These forms avoid the subtraction of covariance matrices and preserve symmetry by construction, but increase the computational cost. As shown in \Cref{app:flop_count}, this increase is modest for the Kalman filter but becomes significant for the RTS smoother. Moreover, the symmetric Joseph form makes it straightforward to derive its square-root implementations, which further improve numerical properties at the expense of additional computational cost. In this work, square-root implementations are not studied further.

A second and more fundamental source of instability arises from the matrix inversions required by the RTS smoother. While filtering only requires the innovation covariance $\mathbf{S}_k$ to be non-singular, RTS smoothing additionally requires the predicted state covariance matrix $\mathbf{P}_{k+1}^-$ to be non-singular. This requirement may be violated in practice when PN methods operate in high-accuracy regimes where the filtered covariance matrix $\mathbf{P}_k$ becomes negligible, such that the predicted covariance matrix $\mathbf{P}_{k+1}^- = \mathbf{A}_k \mathbf{P}_k \mathbf{A}_k^\top + \mathbf{Q}_k$ is dominated by the process noise covariance $\mathbf{Q}_k$. For small time steps, this matrix may become ill-conditioned, potentially leading to numerical instabilities or breakdown of the RTS equations \citep{Kraemer_Stable}.

\section{Inference in State-Space Models through Dynamic Programming}
\label{sec:dynamic_programming}
This section formulates an optimization view of Bayesian smoothing in linear Gaussian state-space models to derive an alternative smoothing algorithm that mitigates the numerical issues of the RTS smoother. Because the smoothing posterior is Gaussian, it is fully characterized by a quadratic negative log-posterior: its minimizer gives the posterior mean trajectory, while its curvature determines the smoothing covariances. Writing smoothing as a maximum-a-posteriori (MAP) estimation problem reveals a dynamic-programming structure \citep{Sarkka_BFS}. The Kalman filter corresponds to the forward elimination of past states, while the backward pass can be interpreted in two equivalent ways. One interpretation leads to the classical RTS smoother through backward substitution; the other leads to the MBF smoother through adjoint equations.

Although this viewpoint is well known \citep[see e.g.][]{Cox_DP}, it provides a unifying framework for interpreting different smoothing algorithms. In particular, it shows that the MBF variables are not merely an algebraic reformulation of RTS quantities, but arise as adjoint variables of the same optimization problem. This interpretation will be central in \Cref{sec:hyperparameter_optimization} where these quantities are reused for gradient-based hyperparameter learning. Full derivations are deferred to \Cref{app:dynamic_programming}.

\subsection{Maximum-A-Posteriori Estimation}
Consider the posterior distribution over the state trajectory. The corresponding MAP estimator is given by:
\begin{align}
    \mathbf{x}_{0:T}^{\text{MAP}} &= \argmax_{\mathbf{x}_{0:T}} p(\mathbf{x}_{0:T} \mid \mathbf{y}_{1:T})\\
    &= \argmax_{\mathbf{x}_{0:T}}p(\mathbf{x}_{0}) \prod_{k=1}^{T}p(\mathbf{y}_k \mid \mathbf{x}_k) p(\mathbf{x}_k \mid \mathbf{x}_{k-1})
\end{align}
Equivalently, this can be written as the minimization of the negative log-posterior  \citep{Sarkka_BFS}: 
\begin{align}\label{eq:MAP}
    \mathbf{x}_{0:T}^{\text{MAP}} = \argmin_{\mathbf{x}_{0:T}}L(\mathbf{x}_{0:T})
\end{align}
where $L(\mathbf{x}_{0:T})$ represents the loss function, defined as:
\begin{align}
    L(\mathbf{x}_{0:T}) =&  -\log \left[p(\mathbf{x}_0) \prod_{k=1}^T p(\mathbf{y}_k | \mathbf{x}_k) p(\mathbf{x}_k | \mathbf{x}_{k-1}) \right] \\ \label{eq:cost_function}
    \begin{split}
        =& -\log p(\mathbf{x}_0) - \sum_{k=1}^T \log p(\mathbf{y}_k |  \mathbf{x}_k) \\
        &\qquad \qquad \quad  - \sum_{k=1}^T \log p(\mathbf{x}_k | \mathbf{x}_{k-1})    
    \end{split}
\end{align}
Under the linear Gaussian state-space model specified in \Cref{sec:inference_in_state_space_models}, the loss function becomes:
\begin{align}\label{eq:L_cost_function}
\begin{split}
    L(\mathbf{x}_{0:T}) = C  + \frac{1}{2} &\| \mathbf{x}_0 - \mathbf{m}_0 \|_{\mathbf{P}_0^{-1}}^2 \\
    + \ \frac{1}{2} \sum_{k=1}^\top&\| \mathbf{y}_k - \mathbf{H}_k \mathbf{x}_k \|_{\mathbf{R}_k^{-1}}^2\\
    + \ \frac{1}{2} \sum_{k=1}^\top  &\| \mathbf{x}_{k} - \mathbf{A}_{k-1} \mathbf{x}_{k-1} \|_{\mathbf{Q}^{-1}_{k-1}}^2
\end{split}
\end{align}
where $C$ is a constant independent of the state trajectory. This is a sparse quadratic optimization problem, where the sparsity is induced by the Markovian structure of the state-space model. As a result, the optimal trajectory and associated Gaussian marginals can be computed efficiently using dynamic programming.

\subsection{Forward Dynamic Programming}
Dynamic programming solves the MAP estimation problem by decomposing it into a sequence of local subproblems. In the forward pass, past states are recursively eliminated while the cost is accumulated. This procedure recovers the predicted and filtered means and covariances as determined by the Kalman filter.

For the first time step, one minimizes the following loss function with respect to $\mathbf{x}_0$ as a function of $\mathbf{x}_1$,
\begin{align}
    \min_{\mathbf{x}_0} \frac{1}{2} \| \mathbf{x}_0 - \mathbf{m}_0 \|_{\mathbf{P}_0^{-1}}^2 + \frac{1}{2} \| \mathbf{x}_1 - \mathbf{A}_0 \mathbf{x}_0 \|_{\mathbf{Q}_0^{-1}}^2
\end{align}
The resulting minimizer expresses $\mathbf{x}_0$ conditionally as a function of $\mathbf{x}_1$, which we denote as the conditional mean $\mathbf{m}_0^c(\mathbf{x}_{1})$. For subsequent steps $k=1\ldots T-1$, a similar elimination step results in the conditional means $\mathbf{m}^c_k(\mathbf{x}_{k+1})$ by minimizing: 
\begin{align}
\begin{split}
    \min_{\mathbf{x}_k} \frac{1}{2} \| \mathbf{y}_k - \mathbf{H}_k \mathbf{x}_k \|_{\mathbf{R}_k^{-1}}^2 + \frac{1}{2} &\| \mathbf{x}_{k+1} - \mathbf{A}_k \mathbf{x}_k \|_{\mathbf{Q}_k^{-1}}^2
    \\
    + \frac{1}{2} &\| \mathbf{x}_{k} - \mathbf{m}^-_{k} \|_{[\mathbf{P}_{k}^-]^{-1}}^2
\end{split}
\end{align}
At the final time step $k=T$, the optimization problem becomes:
\begin{align}
    \min_{\mathbf{x}_T} \frac{1}{2} &\| \mathbf{y}_T - \mathbf{H}_T \mathbf{x}_T \|_{\mathbf{R}_T^{-1}}^2 + \frac{1}{2}\| \mathbf{x}_T - \mathbf{m}^-_{T} \|_{[\mathbf{P}_{T}^-]^{-1}}^2
\end{align}
This uniquely determines the final filtered mean and covariance, which coincide with the final smoothing mean and covariance.

As shown in \Cref{app:dynamic_programming}, this forward elimination procedure recovers the Kalman filtering equations. Hence, it produces the predicted and filtered means and covariances at all time steps. In addition, it provides conditional relations expressing each optimal state $\mathbf{x}_k$ in terms of the next state $\mathbf{x}_{k+1}$, denoted as $\mathbf{m}_k^c(\mathbf{x}_{k+1})$, for $k=0\ldots T-1$. Finally, the minimal value of the loss function can also be determined and reads:
\begin{align}
    \min_{\mathbf{x}_{0:T}} L(\mathbf{x}_{0:T}) = C + \frac{1}{2} \sum_{k=1}^\top \| \mathbf{y}_k - \mathbf{H}_k \mathbf{m}_k^- \|_{\mathbf{S}_{k}^{-1}}^2
\end{align}

\subsection{Backward Substitution}
At the end of the forward dynamic-programming pass, the final state $\mathbf{x}_T$ has been determined and each preceding state $\mathbf{x}_k$ is expressed conditionally in terms of the next one $\mathbf{x}_{k+1}$ through the conditional mean $\mathbf{m}_k^c(\mathbf{x}_{k+1})$, for $k=0\ldots T-1$. The optimal state trajectory can hence be recovered by recursively substituting these conditional relations backwards from $k=T-1$ to $k=0$.

As shown in \Cref{app:dynamic_programming}, this backward substitution procedure recovers the RTS smoothing equations. Thus, RTS smoothing can be interpreted as solving the sparse MAP problem by forward elimination followed by backward substitution.

\subsection{Adjoint Equations}
An alternative but equivalent backward pass can be obtained by considering sensitivities of the optimal cost with respect to the intermediate quantities produced during the forward pass.
The forward pass yields both the minimal cost and the predicted and filtered distributions on which this cost depends. The derivatives of this optimal cost with respect to these intermediate quantities can hence be propagated backwards in time.

Rather than propagating the smoothed means and covariances directly, this approach propagates adjoint variables. These adjoints measure how changes in the intermediate predictive and filtering quantities affect the optimal value of the cost function. This viewpoint is closely related to adjoint methods, Lagrangian duality, and backpropagation. Since the cost function is quadratic, its gradients are linear and contain sufficient information to reconstruct the smoothed state trajectory through the optimality conditions, as demonstrated in \Cref{app:dynamic_programming}.

As further shown in \Cref{app:dynamic_programming}, this adjoint-based backward pass recovers the MBF smoothing equations, which will be detailed in \Cref{sec:modified_bryson_frazier_smoothing}. This interpretation provides a derivation that is independent of the RTS recursions and additionally establishes a direct connection between smoothing and gradient computation. This connection will be leveraged in \Cref{sec:hyperparameter_optimization} to optimize kernel hyperparameters.

\section{Modified Bryson-Frazier Smoother}
\label{sec:modified_bryson_frazier_smoothing}
This section details the MBF smoothing equations, summarizing their algorithmic form following the qualitative and quantitative derivations in \Cref{sec:dynamic_programming} and \Cref{app:dynamic_programming}, respectively.

\subsection{Modified Bryson-Frazier Smoothing}
The MBF smoother reformulates the backward smoothing pass in terms of adjoint variables rather than propagating the smoothed means and covariances directly. Specifically, it propagates two sets of variables, $(\bm{\lambda}^+_k, \mathbf{\Lambda}^+_k)$ and $(\bm{\lambda}^-_k, \mathbf{\Lambda}^-_k)$, in a backward pass similar to the RTS smoother. The former are associated with the filtered state at time $k$, while the latter are associated with the predicted state at time $k$. Starting from $\bm{\lambda}_T^+ = \mathbf{0}$ and $\mathbf{\Lambda}_T^+ = \mathbf{0}$, the backward pass proceeds for $k=T\ldots1$. First, an update step computes
\begin{align}
    \bm{\lambda}_k^- &= (\mathbf{I} - \mathbf{K}_k \mathbf{H}_k)^\top \bm{\lambda}_k^+ - \mathbf{H}_k^\top \mathbf{S}_k^{-1} \mathbf{z}_k \\
    \mathbf{\Lambda}_k^- &= (\mathbf{I} - \mathbf{K}_k \mathbf{H}_k)^\top \mathbf{\Lambda}_k^+ (\mathbf{I} - \mathbf{K}_k \mathbf{H}_k) + 
    \mathbf{H}_k^\top \mathbf{S}_k^{-1} \mathbf{H}_k
\end{align}
which is followed by a prediction step:
\begin{align}
    \bm{\lambda}_{k-1}^+ &= \mathbf{A}_{k-1}^\top \bm{\lambda}_{k}^- \\
    \mathbf{\Lambda}_{k-1}^+ &= \mathbf{A}_{k-1}^\top \mathbf{\Lambda}_k^- \mathbf{A}_{k-1}
\end{align}
The smoothed variables can then be recovered, either from $(\bm{\lambda}^+_k, \mathbf{\Lambda}^+_k)$, using,
\begin{align}
    \mathbf{m}_{k}^s &= \mathbf{m}_{k} - \mathbf{P}_{k} \bm{\lambda}_k^+ \\
    \mathbf{P}_{k}^s &= \mathbf{P}_{k} - \mathbf{P}_{k} \mathbf{\Lambda}_k^+ \mathbf{P}_{k}
\end{align}
or from  $(\bm{\lambda}^-_k, \mathbf{\Lambda}^-_k)$, using,
\begin{align}
    \mathbf{m}_{k}^s &= \mathbf{m}_{k}^- - \mathbf{P}_{k}^- \bm{\lambda}_k^- \\
    \mathbf{P}_{k}^s &= \mathbf{P}_{k}^- - \mathbf{P}_{k}^- \mathbf{\Lambda}_k^- \mathbf{P}_{k}^-
\end{align}
Both retrieval formulas result in the same smoothed means and covariances as the RTS smoother whenever the latter is well-defined. Indeed, \citet{Bierman1973} originally derived the MBF algorithm as a reformulation of the RTS smoother in continuous time. For completeness, a self-contained derivation of the MBF smoother from the RTS recursions in discrete-time is provided in \Cref{sec:app:MBF-from-RTS}.

\subsection{Numerical Considerations}
From a numerical perspective, the MBF smoother has several favorable properties in terms of numerical stability.
First, the adjoint variables propagated by the MBF smoother do not involve differences of covariance matrices, which eliminates the risk of losing positive definiteness through the accumulation of round-off errors. Although covariance differences reappear in the reconstruction of the smoothed covariances, they are not themselves propagated through the backward recursion, preventing the accumulation of such errors. Moreover, these remaining subtractions can be eliminated using Householder transformations and square-root formulations \citep{square_root_MBF}.
Second, the MBF smoother removes the additional non-singularity requirement imposed by the RTS smoother on the predicted state covariance matrix. Indeed, the smoothing pass only requires the same matrix inversions as the Kalman filter, namely those involving the innovation covariance matrices $\mathbf{S}_k$. Finally, the MBF smoother also reduces computational and memory requirements as compared to the RTS smoother, as quantified in \Cref{app:flop_count}, which provides a detailed flop-count and memory analysis of the algorithms discussed in this paper.

\section{Hyperparameter Optimization}
\label{sec:hyperparameter_optimization}
Kernel functions in Gaussian processes typically depend on unknown parameters, often referred to as hyperparameters. A common approach to estimate these is to maximize the marginal likelihood, also known as the evidence \citep{Rasmussen_GP}. In standard GP regression, evaluating and differentiating the marginal likelihood scales cubically in the number of data points. In the state-space formulation, however, both quantities can be computed with linear scaling in the number of time points \citep{Sarkka_BFS, Parellier_Speeding_Up_Backprop}. This section shows that the MBF smoother naturally provides part of the quantities required for gradient-based hyperparameter optimization. The remaining quantities can be determined through an auxiliary backward recursion.

Consider the negative log marginal likelihood (NLL):
\begin{align}
    J_{1:T}^{\text{NLL}}
    &= - \log p(\mathbf{y}_{1:T} \mid \theta)
\end{align}
where $\theta$ denotes a hyperparameter. Due to the Markovian structure, it can be expressed, up to a constant, using the prediction error decomposition as \citep{schweppe1965evaluation}: 
\begin{align}\label{eq:JNLL_cost}
   J_{1:T}^{\text{NLL}} &= \frac{1}{2}\sum_{k=1}^T \left(\| \mathbf{y}_k - \mathbf{H}_k \mathbf{m}_{k}^- \|_{\mathbf{S}_k^{-1}}^2 + \log \det(\mathbf{S}_k) \right)
\end{align}
We denote the data-fit term by:
\begin{align}
    J_{1:T} = \frac{1}{2} \sum_{k=1}^T \| \mathbf{y}_k - \mathbf{H}_k \mathbf{m}_{k}^- \|_{\mathbf{S}_k^{-1}}^2
\end{align}
For regression, the innovation covariance $\mathbf{S}_k$ does not depend on the predicted or filtered means $\mathbf{m}_k^-$ and $\mathbf{m}_k$. As a result, the derivatives of the NLL with respect to these means coincide with those of $J_{1:T}$:
\begin{align}
    &\frac{\partial J^{\text{NLL}}_{1:T}}{\partial \mathbf{m}_{k}^-} = \frac{\partial J^{\text{NLL}}_{k:T}}{\partial \mathbf{m}_{k}^-}\ \, = \ \ \frac{\partial J_{k:T}}{\partial \mathbf{m}_{k}^-} \ = \frac{\partial J_{1:T}}{\partial \mathbf{m}_{k}^-} \\
    &\frac{\partial J^{\text{NLL}}_{1:T}}{\partial \mathbf{m}_{k}} = \frac{\partial J^{\text{NLL}}_{k+1:T}}{\partial \mathbf{m}_{k}} = \frac{\partial J_{k+1:T}}{\partial \mathbf{m}_{k}} = \frac{\partial J_{1:T}}{\partial \mathbf{m}_{k}}
\end{align}

\subsection{Adjoint View on MBF Smoothing}
As shown in \Cref{app:dynamic_programming}, the MBF smoother can be derived as an adjoint method applied to the MAP objective. Through the equivalence above, it can hence be shown that, for regression, these adjoint quantities coincide with the derivatives of the NLL with respect to the predicted and filtered means:
\begin{align}
    \bm{\lambda}_k^- &= \frac{\partial J_{k:T}^{\text{NLL}}}{\partial \mathbf{m}_{k}^-}&&\bm{\lambda}_k^+ = \frac{\partial J_{k+1:T}^\text{NLL}}{\partial \mathbf{m}_{k}} \\
    \mathbf{\Lambda}_k^- &= \frac{\partial^2 J_{k:T}^{\text{NLL}}}{\partial (\mathbf{m}_{k}^-)^2} &&\mathbf{\Lambda}_k^+ = \frac{\partial^2 J_{k+1:T}^\text{NLL}}{\partial (\mathbf{m}_{k})^2}
\end{align}
Thus, this establishes the result that the MBF smoother computes the first- and second-order derivatives of the NLL with respect to the intermediate mean quantities produced by the Kalman filter.

\subsection{Gradients with Respect to Covariances}
To obtain gradients of the NLL with respect to kernel hyperparameters, derivatives with respect to the predicted and filtered covariance matrices are also required. These derivatives are not directly provided by the MBF smoothing equations, but can be computed through an auxiliary backward recurrence. Define:
\begin{align}\label{eq:def_T}
    \mathbf{T}_k^- &= \frac{\partial J_{1:T}^{\text{NLL}}}{\partial \mathbf{P}_{k}^-} \\
    \mathbf{T}_k^+ &= \frac{\partial J_{1:T}^{\text{NLL}}}{\partial \mathbf{P}_{k}} 
\end{align}
and correction factors, ($\widetilde{\mathbf{\Lambda}}_k^-, \widetilde{\mathbf{\Lambda}}_k^+$), such that the covariance sensitivities admit the following decomposition:
\begin{align}
    \mathbf{T}_k^- &= \mathbf{\Lambda}_k^- + \widetilde{\mathbf{\Lambda}}_k^- \\ \label{eq:decomposition}
    \mathbf{T}_k^+ &= \mathbf{\Lambda}_k^+ + \widetilde{\mathbf{\Lambda}}_k^+
\end{align}
As detailed in \Cref{app:correction_equations}, the correction terms can be computed through an auxiliary backward recurrence. Starting from  $\widetilde{\mathbf{\Lambda}}_T^+ = \boldsymbol{0}$, for $k=T, \ldots, 1$, it reads:
\begin{align}\label{eq:corr_rec_1}
    \mathbf{l}_k &={} \frac{1}{2} \bm{\lambda}_k^+ \mathbf{z}_k^\top \mathbf{R}_k^{-1} \mathbf{H}_k \\\label{eq:corr_rec_2}
    \begin{split}
    \widetilde{\mathbf{\Lambda}}_k^- 
    &={} (\mathbf{I} - \mathbf{K}_k \mathbf{H}_k)^\top \widetilde{\mathbf{\Lambda}}_k^+ (\mathbf{I} - \mathbf{K}_k \mathbf{H}_k) \\
    &+(\mathbf{I} - \mathbf{K}_k \mathbf{H}_k)^\top (\mathbf{l}_k + \mathbf{l}_k^\top) (\mathbf{I} - \mathbf{K}_k \mathbf{H}_k)\\
    &{}-\frac{1}{2} \mathbf{H}_k^\top \mathbf{S}_k^{-1} \mathbf{H}_k - \frac{1}{2}\mathbf{H}_k^\top \mathbf{S}_k^{-1} \mathbf{z}_k \mathbf{z}_k^\top \mathbf{S}_k^{-1} \mathbf{H}_k
    \end{split}
    \\\label{eq:corr_rec_3}
    \widetilde{\mathbf{\Lambda}}_{k-1}^+ &={} \mathbf{A}_{k-1}^\top \widetilde{\mathbf{\Lambda}}_k^- \mathbf{A}_{k-1}
\end{align}
Together with the MBF variables, these correction terms provide the full derivatives of the NLL with respect to the predicted and filtered covariance matrices.

\subsection{Kernel Hyperparameter Optimization}
Having access to derivatives with respect to both means and covariances enables efficient computation of gradients of the NLL with respect to the hyperparameter $\theta$. Assuming the dependence on $\theta$ enters through the state-space matrices $\mathbf{A}_k$ and $\mathbf{Q}_k$, and the initial covariance $\mathbf{P}_0$, the chain rule gives,
\begin{align}\label{eq:gradient_expression_conceptual}
\begin{split}
    \frac{\text{d}J_{1:T}^{\text{NLL}}}{\text{d}\theta} = \sum_{k=0}^{T-1} & \left(\left\langle \frac{\partial J_{1:T}^{\text{NLL}}}{\partial \mathbf{A}_k}, \frac{\text{d} \mathbf{A}_k}{\text{d}\theta} \right\rangle + \left\langle \frac{\partial J_{1:T}^{\text{NLL}}}{\partial \mathbf{Q}_{k}}, \frac{\text{d} \mathbf{Q}_{k}}{\text{d}\theta} \right\rangle \right) \\
    &+ \left\langle \frac{\partial J_{1:T}^{\text{NLL}}}{\partial \mathbf{P}_0}, \frac{\text{d}\mathbf{P}_0}{\text{d}\theta} \right\rangle
\end{split}
\end{align} 
where $\langle \cdot, \cdot \rangle$ denotes the Frobenius inner product. Each term can be expanded further; the corresponding derivations are deferred to \cref{app:gradient_expr}. This leads to the following expression for the gradient of the NLL with respect to a hyperparameter $\theta$:
\begin{align}\label{eq:gradient_expr}
\begin{split}
    \frac{\text{d}J_{1:T}^{\text{NLL}}}{\text{d}\theta} &= \sum_{k=0}^{T-1} \left\langle ( \bm{\lambda}_{k+1}^- \mathbf{m}_{k}^\top + 2\mathbf{T}_{k+1}^- \mathbf{A}_k \mathbf{P}_{k}), \frac{\text{d}\mathbf{A}_k}{\text{d}\theta} \right\rangle \\
    &+ \sum_{k=0}^{T-1} \left \langle \mathbf{T}_{k+1}^-, \frac{\text{d}\mathbf{Q}_k}{\text{d}\theta} \right \rangle + \left \langle \mathbf{T}_0^+,\frac{\text{d}\mathbf{P}_0}{\text{d}\theta} \right\rangle 
\end{split}
\end{align}
This expression enables gradient-based hyperparameter optimization using quantities obtained from the Kalman filter, the MBF smoother, and the additional covariance-correction recursion.

\section{Numerical Experiments}
\label{sec:numerical_experiments}
We demonstrate the MBF smoother through three numerical experiments that highlight its correctness, practical usefulness, and robustness. First, we show that GP regression can be performed using Kalman filtering and MBF smoothing, yielding results identical to conventional RTS smoothing in a well-conditioned setting. Second, we illustrate how the intermediate quantities propagated by the MBF smoother enable efficient kernel hyperparameter optimization. Finally, we show that MBF smoothing remains applicable in singular state-space models where RTS smoothing breaks down. The code for the numerical experiments is available online.\footnote{\scriptsize \href{https://github.com/TomColemont/ProbNum26-MBF-Smoothing}{github.com/TomColemont/ProbNum26-MBF-Smoothing}}

\subsection{Experimental Set-Up}
We generate a latent function by drawing from a Gaussian process with a Matérn kernel with smoothness parameter $\nu = 3/2$,
\begin{align}
    k(t, t') = \sigma_f^2 (1 + \lambda |t-t'|) \exp(-\lambda|t-t'|)
\end{align}
where $\sigma_f=1$ is the output scale and $\lambda = 1$, corresponding to the length scale $\ell = \sqrt{3}/\lambda = \sqrt{3}$ \citep{Sarkka_ASDE}. In total, $5\,000$ data points are generated over a time-interval of $200$ seconds, of which $200$ data points are uniformly selected as training data. 

Following \citet{Sarkka_SSM, Matern_DiscreteTimeMatrices}, this GP can be represented in discrete-time as a linear Gaussian state-space model:
\begin{align}
    \Delta t_k &:= t_{k+1} - t_k \\
    \phi_k &:= \exp(-\lambda \Delta t_k) \\
    \mathbf{A}_k &= \phi_k\begin{bmatrix}
        1 + \lambda \Delta t_k & \Delta t_k \\
        - \lambda^2 \Delta t_k & 1 - \lambda \Delta t_k
    \end{bmatrix}\\
    [\mathbf{Q}_k]_{11} &= \sigma_f^2 (1 - \phi_k^2 (1 + 2 \lambda \Delta t_k + 2 (\lambda \Delta t_k)^2)) \\
    [\mathbf{Q}_k]_{12} &= [\mathbf{Q}_k]_{21} = 2 \sigma_f^2 \Delta t_k^2 \lambda^3 \phi_k^2 \\
    [\mathbf{Q}_k]_{22} &= \sigma_f^2 \lambda^2 (1 - \phi_k^2(1 - 2 \lambda \Delta t_k + 2(\lambda \Delta t_k)^2))\\
    \mathbf{H}_k &= \begin{bmatrix}
        1 & 0
    \end{bmatrix} \qquad \mathbf{R}_k = \sigma_n^2 = 10^{-2}
\end{align}
The initial distribution is given by
\begin{align}
    \mathbf{x}_0 \sim \mathcal{N}\left( \mathbf{0}, \ \sigma_f^2 \begin{bmatrix}
        1 & 0 \\
        0 & \lambda^2
    \end{bmatrix} \right)
\end{align}
where $\mathbf{P}_0$ is the stationary prior covariance. We use this set-up to study the MBF smoother in three settings.

\begin{figure}
    \centering
    \includegraphics[width=\linewidth]{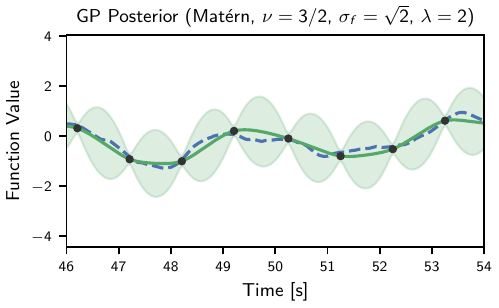}
    \caption{Segment of the latent function (dashed blue), observed at uniformly sampled time instances (black). The GP posterior with mean (solid green) and $2\sigma$-interval (light green) are shown for a Matérn-3/2 kernel with misspecified hyperparameters, resulting in overly large posterior uncertainty.}
    \label{fig:experiment1}
\end{figure}

\subsection{Gaussian Process Regression}\label{subsec:GPR}
\cref{fig:experiment1} shows the GP posterior obtained using a Matérn-3/2 kernel with deliberately misspecified hyperparameters, namely $\sigma_f=\sqrt{2}$ and $\lambda=2$. These differ from the values used to generate the latent function and therefore lead to an intentionally misspecified posterior. In this well-conditioned setting, all state-space covariance matrices remain regular, so both the RTS and MBF smoothers are applicable. As expected, the MBF-smoothed and RTS-smoothed means coincide up to machine precision and match the GP posterior computed through a Cholesky factorization of the full covariance matrix \citep{Rasmussen_GP}. This experiment verifies that the MBF smoother is an exact drop-in replacement for RTS smoothing when the state-space model is non-singular.

\subsection{Hyperparameter Optimization}\label{subsec:hyperparameter_optimization}
As described in \Cref{sec:hyperparameter_optimization}, the quantities propagated by the MBF smoother can be reused to compute the gradients of the negative log marginal likelihood, or equivalently negative log evidence, with respect to kernel hyperparameters. \Cref{fig:experiment2} illustrates the minimization trajectory in the two-dimensional hyperparameter space $(\sigma_f, \lambda)$ when minimizing the negative log evidence using gradient descent with a learning rate of $1.5\cdot10^{-3}$. Other first-order optimizers could be used as well. Lighter colors indicate lower values of the objective function. The resulting hyperparameter estimates are $\hat{\sigma}_f = 1.04$ and $\hat{\lambda} = 0.93$, which closely match the true hyperparameters of the generative model. This experiment demonstrates that the MBF smoother provides the quantities needed for gradient-based kernel hyperparameter optimization. 

\begin{figure}
    \centering
    \includegraphics[width=\linewidth]{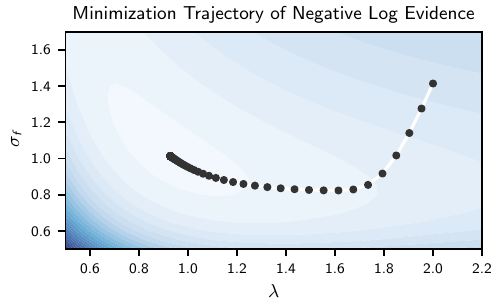}
    \caption{Minimization trajectory of the negative log evidence as a function of the kernel hyperparameters $\sigma_f$ and $\lambda$. The gradients are computed directly from MBF intermediate quantities. Lighter colors indicate lower values of the cost function. Gradient descent converges to $\hat{\sigma}_f = 1.04$ and $\hat{\lambda} = 0.93$, close to the true hyperparameters.}
    \label{fig:experiment2}
\end{figure}

\begin{figure}
    \centering
    \includegraphics[width=\linewidth]{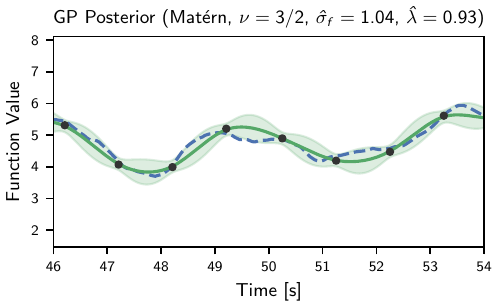}
    \caption{GP posterior with mean (solid green) and $2\sigma$ (light green) after augmenting the state-space model with a deterministic bias term, making both the state and process noise covariances singular. The MBF smoother remains stable and matches the corresponding GP posterior, while the RTS smoother fails due to the singular predicted covariance matrix.}
    \label{fig:experiment3}
\end{figure}

\subsection{Singular State-Space Models}\label{subsec:singular_sssm}
Finally, we consider a state-space model in which both the state and process noise covariance matrices become singular. Such situations arise naturally when augmenting the state vector with deterministic components, such as a constant bias term. We add a bias of $5$ units to the latent function, and include this as a deterministic bias in the state-space model. This introduces zeros in the final row and column of both $\mathbf{P}_k$ and $\mathbf{Q}_k$. \Cref{fig:experiment3} shows the resulting GP posterior when the estimated hyperparameters from \Cref{subsec:hyperparameter_optimization} are used.

In this singular setting, the RTS smoother is no longer applicable because it requires inversion of the predicted state covariance matrix. In contrast, the MBF smoother remains well-defined and produces a smoothed trajectory that agrees with the GP posterior up to machine precision. This illustrates the MBF smoother's robustness in regimes where the RTS smoother breaks down.

\section{Discussion}
\label{sec:discussion_and_conclusion}
The results of this work demonstrate that the modified Bryson-Frazier (MBF) smoother provides a principled alternative to the Rauch-Tung-Striebel (RTS) smoother for Gaussian process regression in its state-space representation. In well-conditioned settings, MBF smoothing is fully equivalent to RTS smoothing and yields the same posterior distributions as standard Gaussian process regression. Its main advantage appears in numerically challenging regimes: by avoiding inversion of the predicted state covariance matrix, MBF smoothing remains applicable when this matrix becomes ill-conditioned or singular. This is particularly relevant in high-accuracy regimes, where covariance matrices may become nearly singular, and in constrained or augmented state-space models, as illustrated in \Cref{subsec:singular_sssm}.

Beyond numerical robustness, MBF smoothing also offers computational advantages. Although both RTS and MBF smoothing are cubic in the state dimension, the latter has lower leading-order computational cost and reduced memory requirements, as quantified in \Cref{app:flop_count}. Moreover, unlike the RTS smoother, it does not propagate covariance differences backwards, avoiding round-off error accumulation that can lead to loss of positive definiteness. The remaining covariance subtractions during the reconstruction can be avoided using square-root implementations \citep{square_root_MBF}.

A useful consequence of the MBF smoothing formulation is that smoothed means and covariances need not be reconstructed at every step of the backward pass. Instead, the algorithm propagates adjoint variables, from which smoothed quantities can be recovered only when needed. This is especially relevant in probabilistic ODE solvers, where small internal time steps may be required in stiff regions even when the solution is only requested at a coarser set of output times.

The adjoint interpretation also connects MBF smoothing to kernel hyperparameter optimization. The MBF backward pass naturally provides derivatives of the negative log marginal likelihood with respect to the predicted and filtered means. Together with an additional covariance-correction recursion, this yields the sensitivities required for hyperparameter optimization, allowing smoothing and hyperparameter learning to be treated within a single state-space framework.

This work has focused on one-dimensional Gaussian process regression with stationary, integrable kernels, for which exact or arbitrarily accurate state-space representations are available. Natural extensions include spatio-temporal Gaussian processes, where larger state dimensions may make the computational advantages of MBF smoothing more pronounced \citep{Sarkka_spatiotemporal}, and probabilistic ODE solvers. While MBF smoothing applies directly to the latter class of solvers, MBF-based hyperparameter optimization requires further adaptation for EK1-based solvers, where the innovation covariance depends on the predicted mean. Combining this with iterated extended smoothing schemes is especially interesting, as these methods naturally involve multiple forward-backward passes necessary for gradient-based hyperparameter optimization.

\vspace{-5pt}
\section{Conclusion}
\vspace{-5pt}
\label{sec:conclusion}
We have revisited the modified Bryson-Frazier smoother as a numerically robust and computationally efficient alternative to RTS smoothing for Gaussian process regression in its state-space representation. The MBF smoother computes the same posterior distributions as the RTS smoother whenever both are well-defined, while avoiding the inversion of predicted state covariance matrices that can make RTS smoothing unstable or infeasible. Furthermore, the MBF backward variables admit an adjoint interpretation, allowing them to be reused for gradient-based kernel hyperparameter optimization with minimal additional cost.

Overall, the MBF smoother provides a unified route to stable smoothing and hyperparameter learning in Gaussian processes in their state-space representation, making it a useful building block for probabilistic numerical methods based on Bayesian filtering and smoothing.

\vspace{-5pt}
\section*{Acknowledgments}
\vspace{-5pt}
This research was partially supported by the Research Foundation -- Flanders (FWO; Grant No.\ I000725N, I002123N, and 1178926N).

\vspace{-5pt}
\bibliography{references}

\appendix

\section{MBF-RTS Smoother Equivalence}\label{sec:app:MBF-from-RTS}
\citet{Bierman1973} originally derived the MBF smoother from the RTS smoother recursions by relying on the underlying continuous-time differential equations. In contrast, this section provides a self-contained derivation of the discrete-time MBF smoother. We start by defining the MBF variables explicitly in terms of the predicted, filtered and smoothed quantities:
\begin{align}
    \bm{\lambda}_k^+ &= \mathbf{P}_{k}^{-1} (\mathbf{m}_{k}-\mathbf{m}_{k}^s)\\
    \bm{\lambda}_k^- &= [\mathbf{P}_{k}^-]^{-1} (\mathbf{m}_{k}^--\mathbf{m}_{k}^s)\\
    \mathbf{\Lambda}_k^+ &= \mathbf{P}_{k}^{-1} (\mathbf{P}_{k} - \mathbf{P}_{k}^s) \mathbf{P}_{k}^{-1}\\
    \mathbf{\Lambda}_k^- &= [\mathbf{P}_{k}^-]^{-1} (\mathbf{P}_{k}^- - \mathbf{P}_{k}^s) [\mathbf{P}_{k}^-]^{-1}
\end{align}
These definitions assume that the predicted and filtered state covariance matrices are non-singular. An alternative derivation that derives the MBF smoother without this assumption is provided by \cite{square_root_MBF}.

\subsection{Derivation: Mean Recursions}
After rearranging, the RTS smoothing equation for the mean can be expressed as:
\begin{align}
    \mathbf{m}_{k}^s - \mathbf{m}_{k} &=  \mathbf{J}_k (\mathbf{m}_{k+1}^s - \mathbf{m}_{k+1}^-) \\
    &= \mathbf{P}_{k} \mathbf{A}_k^\top [\mathbf{P}_{k+1}^-]^{-1} (\mathbf{m}_{k+1}^s - \mathbf{m}_{k+1}^-)
\end{align}
Left-multiplication by $\mathbf{P}_{k}^{-1}$ yields:
\begin{align}
    \mathbf{P}_{k}^{-1} (\mathbf{m}_{k}^s - \mathbf{m}_{k}) &= \mathbf{A}_k^\top [\mathbf{P}_{k+1}^-]^{-1} (\mathbf{m}_{k+1}^s - \mathbf{m}_{k+1}^-)\\
    \bm{\lambda}_k^+ &= \mathbf{A}_k^\top \bm{\lambda}_{k+1}^-
\end{align}
which recovers the MBF mean-prediction equation. To derive the update equation, we relate $\bm{\lambda}_k^+$ and $\bm{\lambda}_k^-$ as:
\begin{align}
    \mathbf{P}_{k} \bm{\lambda}_k^+ - \mathbf{P}_{k}^- \bm{\lambda}_k^- &= \mathbf{m}_{k} - \mathbf{m}_{k}^- \\
    &= \mathbf{K}_k \mathbf{z}_k \\
    &= \mathbf{P}_{k}^- \mathbf{H}_k^\top \mathbf{S}_k^{-1} \mathbf{z}_k
\end{align}
Left-multiplication by $[\mathbf{P}_k^-]^{-1}$ yields after rearranging:
\begin{align}
    \bm{\lambda}_k^- &= [\mathbf{P}_{k}^-]^{-1}\mathbf{P}_{k} \bm{\lambda}_k^+ -\mathbf{H}_k^\top \mathbf{S}_k^{-1} \mathbf{z}_k
\end{align}
Using the symmetry of covariance matrices, the first term can be re-expressed as:
\begin{align}
    [\mathbf{P}_{k}^-]^{-1} \mathbf{P}_{k} &= [\mathbf{P}_{k}^-]^{-1} \mathbf{P}_{k}^- (\mathbf{I} - \mathbf{K}_k \mathbf{H}_k)^\top \\ \label{eq:app_identity}
    &= (\mathbf{I} - \mathbf{K}_k \mathbf{H}_k)^\top
\end{align}
which recovers the MBF mean-update equation:
\begin{align}
    \bm{\lambda}_k^- &=  (\mathbf{I} - \mathbf{K}_k \mathbf{H}_k)^\top\bm{\lambda}_k^+ -\mathbf{H}_k^\top \mathbf{S}_k^{-1} \mathbf{z}_k
\end{align}

\subsection{Derivation: Covariance Recursions}
After rearranging, the RTS smoothing equation for the covariance can be expressed as:
\begin{align}
    \mathbf{P}_{k}^s - \mathbf{P}_{k} &= \mathbf{J}_k (\mathbf{P}_{k+1}^s - \mathbf{P}_{k+1}^-) \mathbf{J}_k^\top \\
    &= \mathbf{P}_{k} \mathbf{A}_k^\top [\mathbf{P}_{k+1}^-]^{-1} (\mathbf{P}_{k+1}^s - \mathbf{P}_{k+1}^-) [\mathbf{P}_{k+1}^-]^{-1} \mathbf{A}_k \mathbf{P}_{k}
\end{align}
Left- and right-multiplication by $\mathbf{P}_{k}^{-1}$ recovers the MBF covariance-prediction equation:
\begin{align}
    \mathbf{\Lambda}_k^+ &= \mathbf{A}_k^\top \mathbf{\Lambda}_{k+1}^- \mathbf{A}_k
\end{align}
For the update equation, we relate $\mathbf{\Lambda}_k^+$ and $\mathbf{\Lambda}_k^-$ as:
\begin{align}
    \mathbf{P}_{k} \mathbf{\Lambda}_k^+ \mathbf{P}_{k} - \mathbf{P}_{k}^- \mathbf{\Lambda}_k^- \mathbf{P}_{k}^- &= \mathbf{P}_{k} - \mathbf{P}_{k}^- \\
    &= - \mathbf{K}_k \mathbf{H}_k \mathbf{P}_{k}^-
\end{align}
Left- and right-multiplication by $[\mathbf{P}_{k}^-]^{-1}$ leads to:
\begin{align}
    [\mathbf{P}_{k}^-]^{-1} \mathbf{P}_{k} \mathbf{\Lambda}_k^+ \mathbf{P}_{k} [\mathbf{P}_{k}^-]^{-1} - \mathbf{\Lambda}_k^- &= - [\mathbf{P}_{k}^-]^{-1}\mathbf{K}_k \mathbf{H}_k \\
    &= - \mathbf{H}_k^\top \mathbf{S}_k^{-1} \mathbf{H}_k
\end{align}
Using the identity in \cref{eq:app_identity} recovers the MBF covariance-update equation.
\begin{align}
    \mathbf{\Lambda}_k^- = (\mathbf{I} - \mathbf{K}_k \mathbf{H}_k)^\top \mathbf{\Lambda}_k^+ (\mathbf{I} - \mathbf{K}_k \mathbf{H}_k) + \mathbf{H}_k^\top \mathbf{S}_k^{-1} \mathbf{H}_k
\end{align}

\section{Unified Derivation: Kalman Filter, RTS and MBF Smoother}
\label{app:dynamic_programming}
This Appendix subsequently derives the Kalman filter, RTS and MBF smoother through a unified dynamic programming framework.

\subsection{Kalman Filtering: Forward Dynamic Programming}
We aim to minimize the cost function in \cref{eq:L_cost_function} through forward dynamic programming. Starting from the initial costs:
\begin{align}
    V_0^\star(\mathbf{x}_0) &= \frac{1}{2} \| \mathbf{x}_0 - \mathbf{m}_0 \|_{\mathbf{P}_0^{-1}}^2 \\
    V_1^\star(\mathbf{x}_1) &= \min_{\mathbf{x}_0} \{V_0^\star(\mathbf{x}_0) + \frac{1}{2} \| \mathbf{x}_1 - \mathbf{A}_0 \mathbf{x}_0 \|_{\mathbf{Q}_0^{-1}}^2 \}
\end{align}
We can successively minimize the cost by considering the following problem for $k=1\ldots T-1$:
\begin{align}\label{eq:general_form_Vkp1}
    \begin{split}
        V_{k+1}^\star(\mathbf{x}_{k+1}) 
            ={}&
        \min_{\mathbf{x}_k} \{V_k^\star(\mathbf{x}_k) + \frac{1}{2} \| \mathbf{y}_k - \mathbf{H}_k \mathbf{x}_k \|_{\mathbf{R}_k^{-1}}^2 \\
            &{}+
        \frac{1}{2} \| \mathbf{x}_{k+1} - \mathbf{A}_k \mathbf{x}_k \|_{\mathbf{Q}_k^{-1}}^2 \}
    \end{split}
\end{align}
To end, for $k=T$, the recursion becomes:
\begin{align}
    V_{T+1}^\star &= \min_{\mathbf{x}_T} \{V_T^\star(\mathbf{x}_T) + \frac{1}{2} \| \mathbf{y}_T - \mathbf{H}_T \mathbf{x}_T \|_{\mathbf{R}_T^{-1}}^2 \}
\end{align}
At the end of the forward pass, we both have the exact minimizer $\mathbf{x}_T^\star$ as well as the minimal cost
$V_{T+1}^\star$.
We will show that this forward dynamic programming perspective recovers the well-known Kalman filtering equations. 
To this end, we will exhaustively rely upon the following well-known identity for the sum of two quadratic functions.
\begin{lemma}\label{lem:quad-sum:1}
Let $\mathbf{x}, \mathbf{\bar{x}} \in \mathbb{R}^n$, $\mathbf{b} \in \mathbb{R}^m$, and $\mathbf{\Gamma} \in \mathbb{R}^{m \times n}$.
Suppose that $\mathbf{T}_1 \in \mathbb{R}^{n \times n}$ and $\mathbf{T}_2 \in \mathbb{R}^{m \times m}$ are symmetric positive definite matrices and define the function
\begin{equation}
    q(\mathbf{x}) = \frac{1}{2} \|\mathbf{x} - \mathbf{\bar{x}}\|_{\mathbf{T}_1^{-1}}^2 + \frac{1}{2} \|\mathbf{\Gamma} \mathbf{x} - \mathbf{b}\|_{\mathbf{T}_2^{-1}}^2.
\end{equation}
Then, the following assertions hold.
\begin{enumeratlem}
    \item
    \(
        q(\mathbf{x}) = \frac{1}{2} \|\mathbf{x} - \mathbf{x}^\star\|_{\mathbf{M}^{-1}}^2 + \frac{1}{2} \|\mathbf{\Gamma} \mathbf{\bar{x}} - \mathbf{b}\|_{\mathbf{S}^{-1}}^2
    \),
    where
    \begin{align*}
        \mathbf{S}
            &=
        \mathbf{\Gamma} \mathbf{T}_1 \mathbf{\Gamma}^\top + \mathbf{T}_2, \quad
        \mathbf{L}
            =
        \mathbf{T}_1 \mathbf{\Gamma}^\top \mathbf{S}^{-1},\\
        \mathbf{M}
            &=
        (\mathbf{I} - \mathbf{L} \mathbf{\Gamma}) \mathbf{T}_1,\quad
        \mathbf{x}^\star 
            =
        \mathbf{\bar{x}} + \mathbf{L} (\mathbf{b} - \mathbf{\Gamma} \mathbf{\bar{x}}).
    \end{align*}
    \item
    \(
    \mathbf{x}^\star = \arg\min_{\mathbf{x} \in \mathbb{R}^n} q(\mathbf{x})
    \)
    and 
    \(
    q(\mathbf{x}^\star) = \frac{1}{2} \|\mathbf{\Gamma} \mathbf{\bar{x}} - \mathbf{b}\|_{\mathbf{S}^{-1}}^2.
    \)
\end{enumeratlem}
\end{lemma}

\begin{proof}
    Follows directly from the optimality conditions and the Woodbury matrix identity, see e.g. \citep[Example 1.1]{Rawlings2017Model}.
\end{proof}

First, we will show for all $k = 1 \ldots T$ that:
\begin{align}\label{eq:induction_claim}
\begin{split}
    V_k^\star(\mathbf{x}_k)
        ={}&
    \frac{1}{2} \| \mathbf{x}_k - \mathbf{m}_{k}^- \|_{[\mathbf{P}_{k}^-]^{-1}}^2
    {}+ \underbrace{\frac{1}{2} \sum_{i=1}^{k-1} \| \mathbf{y}_i - \mathbf{H}_i \mathbf{m}_{i}^- \|_{\mathbf{S}_i^{-1}}^2}_{\equiv J_{1:k-1}}\\
        &{}+
    \underbrace{\min_{\mathbf{x}_{k-1}} \left\{ \frac12\|\mathbf{x}_{k-1} - \mathbf{m}_{k-1}^c(\mathbf{x}_k) \|_{[\mathbf{P}_{k-1}^c]^{-1}}^2 \right\}}_{=0}
\end{split}
\end{align}
up to constants, where:
\begin{align}
    \mathbf{m}_{k}^- &= \mathbf{A}_{k-1} \mathbf{m}_{k-1} \\
    \mathbf{P}_{k}^- &= \mathbf{A}_{k-1} \mathbf{P}_{k-1} \mathbf{A}_{k-1}^\top + \mathbf{Q}_{k-1} \\
    \mathbf{S}_k &= \mathbf{H}_k \mathbf{P}_{k}^- \mathbf{H}_k^\top + \mathbf{R}_k \\
    \mathbf{K}_k &= \mathbf{P}_{k}^- \mathbf{H}_k^\top \mathbf{S}_k^{-1} \\
    \mathbf{m}_{k} &= \mathbf{m}_{k}^- + \mathbf{K}_k (\mathbf{y}_k - \mathbf{H}_k \mathbf{m}_{k}^-) \\
    \mathbf{P}_{k} &= (\mathbf{I} - \mathbf{K}_k \mathbf{H}_k) \mathbf{P}_{k}^- \\
    \mathbf{J}_k &= \mathbf{P}_{k} \mathbf{A}_k^\top [\mathbf{P}_{k+1}^-]^{-1}\\
    \mathbf{m}_k^c(\mathbf{x}_{k+1}) &= \mathbf{m}_{k} + \mathbf{J}_k (\mathbf{x}_{k+1} - \mathbf{m}_{k+1}^-) \\
    \mathbf{P}_k^c &= (\mathbf{I} - \mathbf{J}_k \mathbf{A}_k)\mathbf{P}_k =
    \mathbf{P}_k -\mathbf{J}_k \mathbf{P}_{k+1}^- \mathbf{J}_k^\top
\end{align}
where $\mathbf{m}_{k}^-$, $\mathbf{m}_{k}$ and $\mathbf{m}_k^c$ denote the predicted, filtered and conditional smoothed means, respectively. Consider the case where $k=1$, for which:
\begin{align}
    V_1^\star(\mathbf{x}_1) = \min_{\mathbf{x}_0} \left\{ \frac{1}{2}\|\mathbf{x}_0 - \mathbf{m}_{0} \|_{\mathbf{P}_{0}^{-1}}^2 + \frac{1}{2} \| \mathbf{x}_1 - \mathbf{A}_0 \mathbf{x}_0 \|_{\mathbf{Q}_0^{-1}}^2 \right\}
\end{align}
Then, it follows directly from
\Cref{lem:quad-sum:1}
that:
\begin{align}
\begin{split}
    V_1^\star(\mathbf{x}_1)
        ={}&
    \frac{1}{2} \| \mathbf{x}_1 - \mathbf{m}_{1}^- \|_{[\mathbf{P}_{1}^-]^{-1}}^2\\
        &{}+
    \min_{\mathbf{x}_0} \left\{ \frac12\|\mathbf{x}_0 - \mathbf{m}_0^c(\mathbf{x}_1) \|_{[\mathbf{P}_0^c]^{-1}}^2 \right\}
\end{split}
\end{align}
which establishes the claim for $k=1$. Now, suppose the claim holds for some $k= 1, \ldots, T-1$. Then, it follows from \cref{eq:general_form_Vkp1} that:
\begin{align}
\begin{split}
    V_{k+1}^\star(\mathbf{x}_{k+1}) = J_{1:k-1} + \min_{\mathbf{x}_k} \{\frac{1}{2} \| \mathbf{x}_k - \mathbf{m}_{k}^- \|_{[\mathbf{P}_{k}^-]^{-1}}^2 \\
    + \frac{1}{2} \| \mathbf{y}_k - \mathbf{H}_k \mathbf{x}_k \|_{\mathbf{R}_k^{-1}}^2 + \frac{1}{2} \| \mathbf{x}_{k+1} - \mathbf{A}_k \mathbf{x}_k \|_{\mathbf{Q}_k^{-1}}^2 \}
\end{split}
\end{align}
Applying \Cref{lem:quad-sum:1} to the first two norms yields:
\begin{align}\label{eq:start_recovery}
\begin{split}
    &V_{k+1}^\star(\mathbf{x}_{k+1}) = \underbrace{J_{1:k-1} + \frac{1}{2} \| \mathbf{y}_k - \mathbf{H}_k \mathbf{m}_{k}^- \|_{\mathbf{S}_k^{-1}}^2}_{\equiv J_{1:k}} \\
    &+ \min_{\mathbf{x}_k} \{\frac{1}{2} \| \mathbf{x}_k - \mathbf{m}_{k} \|_{\mathbf{P}_{k}^{-1}}^2 + \frac{1}{2} \| \mathbf{x}_{k+1} - \mathbf{A}_k \mathbf{x}_k \|_{\mathbf{Q}_k^{-1}}^2 \}
\end{split}
\end{align}
Hence, applying \Cref{lem:quad-sum:1},
it follows that:
\begin{align}
\begin{split}
    V_{k+1}^\star(\mathbf{x}_{k+1})
        ={}&
    \frac{1}{2} \| \mathbf{x}_{k+1} - \mathbf{m}_{k+1}^- \|_{[\mathbf{P}_{k+1}^-]^{-1}}^2 + J_{1:k}\\
        &{}+
    \min_{\mathbf{x}_{k}} \left\{ \frac12\|\mathbf{x}_{k} - \mathbf{m}_{k}^c(\mathbf{x}_{k+1}) \|_{[\mathbf{P}_{k}^c]^{-1}}^2 \right\}
\end{split}
\end{align}
Consequently, by induction, we have established that \cref{eq:induction_claim} holds for $k=1\ldots T-1$. Consider the case where $k=T$, for which:
\begin{align}
\begin{split}
    V_{T+1}^\star &= J_{1:T-1} + \min_{\mathbf{x}_T} \{\frac{1}{2} \| \mathbf{x}_T - \mathbf{m}_{T}^- \|_{[\mathbf{P}_{T}^-]^{-1}}^2 \\
    &+ \frac{1}{2} \| \mathbf{y}_T - \mathbf{H}_T \mathbf{x}_T \|_{\mathbf{R}_T^{-1}}^2
\end{split}
\end{align}
Then, applying 
\Cref{lem:quad-sum:1}
leads to:
\begin{align}
    \begin{split}
    V_{T+1}^\star ={}& J_{1:T-1} + \frac{1}{2} \| \mathbf{y}_T - \mathbf{H}_T \mathbf{m}_{T}^- \|_{\mathbf{S}_T^{-1}}^2 \\
    &{}+ \min_{\mathbf{x}_T} \{ \frac{1}{2} \| \mathbf{x}_T - \mathbf{m}_{T} \|_{\mathbf{P}_{T}^{-1}}^2\} = J_{1:T}
    \end{split}
\end{align}
It is hence clear that the obtained minimal cost reads:
\begin{align}
    L_{\text{min}}
        \equiv
    \min_{\mathbf{x}_{0:T}} L(\mathbf{x}_{0:T})
        {}={}
    V_{T+1}^\star + C 
        {}={}
    J_{1:T} + C
\end{align}

\subsection{RTS Smoothing: Backsubstitution}
The forward dynamic programming derivation provides not only the minimal cost but also the first and second moments of the predicted, filtered states and conditional smoothed states.
This structure naturally induces a backward recursion that will recover the RTS smoother.
Indeed, it has been shown in a recursive manner that
\begin{align*}
\begin{split}
    L(\mathbf{x}_{0:T})
        ={}&
    C
        +
    J_{1:T}
        +
    \frac{1}{2} \| \mathbf{x}_T - \mathbf{m}_{T} \|_{\mathbf{P}_{T}^{-1}}^2\\
        &+ 
    \frac12\sum_{k=0}^{T-1} \|\mathbf{x}_{k} - \mathbf{m}_{k}^c(\mathbf{x}_{k+1}) \|_{[\mathbf{P}_{k}^c]^{-1}}^2
\end{split}
\end{align*}

This decomposition shows that the posterior consists of a terminal quadratic term and a series of conditional quadratic terms coupling $\mathbf{x}_k$ to $\mathbf{x}_{k+1}$. Interpreting $L$ as a negative log-posterior, this directly corresponds to 
\begin{align}
    p(\mathbf{x}_{T} \mid \mathbf{y}_{1:T}) &{}= \mathcal{N}(\mathbf{m}_{T}, \mathbf{P}_{T})\\
    p(\mathbf{x}_k \mid \mathbf{x}_{k+1}, \mathbf{y}_{1:T}) 
        &{}= \mathcal{N}(\mathbf{m}_k^c(\mathbf{x}_{k+1}), \mathbf{P}_k^c)
\end{align}

Starting from $p(\mathbf{x}_T \mid y_{1:T})$, the marginals follow via the backward recursion
\begin{align}
    p(\mathbf{x}_k \mid \mathbf{y}_{1:T}) 
        =
    \int p(\mathbf{x}_k \mid \mathbf{x}_{k+1}, \mathbf{y}_{1:T} ) 
        p(\mathbf{x}_{k+1} \mid \mathbf{y}_{1:T}) \,\mathrm{d}\mathbf{x}_{k+1},
\end{align}
which, using Gaussian identities, leads to
\begin{align*}
    p(\mathbf{x}_{k} \mid \mathbf{y}_{1:T}) 
        = \mathcal{N}(\mathbf{m}_{k}^s, \mathbf{P}_{k}^s), 
    \quad \forall k = 0 \hdots T-1,
\end{align*}
where
\begin{align}
    \mathbf{m}_{k}^s 
        &= \mathbf{m}_{k} + \mathbf{J}_k (\mathbf{m}_{k+1}^s - \mathbf{m}_{k+1}^-) \\
    \mathbf{P}_{k}^s 
        &= \mathbf{P}_k^c + \mathbf{J}_k \mathbf{P}_{k+1}^s \mathbf{J}_k^\top \\
        &= \mathbf{P}_{k} + \mathbf{J}_k (\mathbf{P}_{k+1}^s - \mathbf{P}_{k+1}^-) \mathbf{J}_k^\top
\end{align}
which recover the RTS smoothing equations.

\subsection{MBF Smoothing: Adjoint Equations}
As shown in \Cref{sec:app:MBF-from-RTS}, the MBF smoother can be derived directly from the RTS smoother. Here, we consider an alternative derivation that does not start from the RTS smoothing equations. Instead, we interpret the MBF variables as adjoints of the dynamic programming objective introduced earlier.

After the forward dynamic programming pass, the minimum value of the cost function can be determined as:
\begin{align}
    L_{\min} = \frac{1}{2} \sum_{i=1}^\top \| \mathbf{y}_i - \mathbf{H}_i \mathbf{m}_{i}^- \|_{\mathbf{S}_i^{-1}}^2 + C
\end{align}
Ignoring the constant $C$, define the corresponding cost-to-go as:
\begin{align}
    J_{k:T}
        \equiv
    \frac{1}{2} \sum_{i=k}^\top \| \mathbf{y}_i - \mathbf{H}_i \mathbf{m}_{i}^- \|_{\mathbf{S}_i^{-1}}^2
        = J_{1:T} - J_{1:k-1}
\end{align}
Define the adjoint variables as the first derivatives of this cost with respect to the predicted and filtered means:
\begin{align}
    \bm{\lambda}_k^- &= \frac{\partial J_{k:T}}{\partial \mathbf{m}_{k}^-} \\ \label{eq:lambda_plus_definition}
    \bm{\lambda}_k^+ &= \frac{\partial J_{k+1:T}}{\partial \mathbf{m}_{k}}
\end{align}
These definitions reflect that $\mathbf{m}_k^-$ affects the current and future terms, whereas $\mathbf{m}_k$ only affects future terms.

We first focus on the mean-prediction equation. Assuming $\bm{\lambda}_{k+1}^+$ and $\bm{\lambda}_{k+1}^-$ are known, then $\bm{\lambda}_k^+$ can be determined through the chain rule:
\begin{align}
    \underbrace{\frac{\partial J_{k+1:T}}{\partial \mathbf{m}_{k}}}_{=\bm{\lambda}_k^+}&=\underbrace{\left( \frac{\partial \mathbf{m}_{k+1}^-}{\partial \mathbf{m}_{k}} \right)^\top}_{=\mathbf{A}_k^\top}\underbrace{\frac{\partial J_{k+1:T}}{\partial \mathbf{m}_{k+1}^-}}_{=\bm{\lambda}_{k+1}^-}
\end{align}
which recovers the MBF mean-prediction equation. We now turn to the mean-update equation. To determine $\bm{\lambda}_k^-$, consider:
\begin{align}
    \frac{\partial J_{k+1:T}}{\partial \mathbf{m}_{k}^-}&=\underbrace{\left( \frac{\partial \mathbf{m}_{k}}{\partial \mathbf{m}_{k}^-} \right)^\top}_{=(\mathbf{I} - \mathbf{K}_k \mathbf{H}_k)^\top}\underbrace{\frac{\partial J_{k+1:T}}{\partial \mathbf{m}_{k}}}_{=\bm{\lambda}_k^+}
\end{align}
where the left-hand side can be expressed in terms of $\bm{\lambda}_k^-$ by rewriting the cost as:
\begin{align}\label{eq:cost_relation}
    J_{k+1:T}=J_{k:T} - \frac{1}{2} \| \mathbf{y}_k - \mathbf{H}_k \mathbf{m}_{k}^- \|_{\mathbf{S}_k^{-1}}^2
\end{align}
which after taking the first derivative results in:
\begin{align}
    \frac{\partial J_{k+1:T}}{\partial \mathbf{m}^-_{k}}=\underbrace{\frac{\partial J_{k:T}}{\partial \mathbf{m}^-_{k}}}_{=\bm{\lambda}_k^-}+\mathbf{H}_k^\top \mathbf{S}_k^{-1} \underbrace{(\mathbf{y}_k - \mathbf{H}_k \mathbf{m}_{k}^-)}_{=\mathbf{z}_k}
\end{align}
Combining both expressions recovers the MBF mean-update equation:
\begin{align}
    \bm{\lambda}_k^-&=(\mathbf{I} - \mathbf{K}_k \mathbf{H}_k)^\top \bm{\lambda}_k^+ - \mathbf{H}_k^\top \mathbf{S}_k^{-1} \mathbf{z}_k
\end{align}
Together with the MBF mean-prediction equation, the first-order adjoint equations are therefore:
\begin{align}
    \bm{\lambda}_k^+&=\mathbf{A}_k^\top \bm{\lambda}_{k+1}^- \\
    \bm{\lambda}_k^-&=(\mathbf{I} - \mathbf{K}_k \mathbf{H}_k)^\top \bm{\lambda}_k^+ - \mathbf{H}_k^\top \mathbf{S}_k^{-1} (\mathbf{y}_k - \mathbf{H}_k \mathbf{m}_{k}^-)
\end{align}
which, by definition of $J_{k:T}$, is initialized from:
\begin{align}
    \bm{\lambda}_T^+=0
\end{align}
The same argument applies to the second derivatives. Define the adjoint variables as:
\begin{align}
    \mathbf{\Lambda}_k^-&=\frac{\partial^2 J_{k:T}}{\partial (\mathbf{m}_{k}^-)^2} \\
    \mathbf{\Lambda}_k^+&=\frac{\partial^2 J_{k+1:T}}{\partial (\mathbf{m}_{k})^2}
\end{align}
We first focus on the covariance-prediction equation. Again, assuming $\mathbf{\Lambda}_{k+1}^+$ and $\mathbf{\Lambda}_{k+1}^-$ are known, then, $\mathbf{\Lambda}_k^+$ can be determined through the chain rule:
\begin{align}
    \underbrace{\frac{\partial^2 J_{k+1:T}}{\partial (\mathbf{m}_{k})^2}}_{=\mathbf{\Lambda}_k^+}=\underbrace{\left( \frac{\partial \mathbf{m}_{k+1}^-}{\partial \mathbf{m}_{k}} \right)^\top}_{=\mathbf{A}_k^\top}\underbrace{\frac{\partial^2 J_{k+1:T}}{\partial (\mathbf{m}_{k+1}^-)^2}}_{=\mathbf{\Lambda}_{k+1}^-}\underbrace{\left( \frac{\partial \mathbf{m}_{k+1}^-}{\partial \mathbf{m}_{k}} \right)}_{=\mathbf{A}_k}
\end{align}
which recovers the MBF covariance-prediction equation. Focusing on the covariance-update equation to determine $\mathbf{\Lambda}_k^-$, consider:
\begin{align}
    \frac{\partial^2 J_{k+1:T}}{\partial (\mathbf{m}_{k}^-)^2}=\underbrace{\left( \frac{\partial \mathbf{m}_{k}}{\partial \mathbf{m}_{k}^-} \right)^\top}_{=(\mathbf{I} - \mathbf{K}_k \mathbf{H}_k)^\top} \underbrace{\left( \frac{\partial^2 J_{k+1:T}}{\partial (\mathbf{m}_{k})^2} \right)}_{=\mathbf{\Lambda}_k^+}\underbrace{\left( \frac{\partial \mathbf{m}_{k}}{\partial \mathbf{m}_{k}^-} \right)}_{=(\mathbf{I} - \mathbf{K}_k \mathbf{H}_k)}
\end{align}
where the left-hand side can be expressed in terms of $\mathbf{\Lambda}_k^-$ in a similar manner as above using \cref{eq:cost_relation}:
\begin{align}
    \frac{\partial^2 J_{k+1:T}}{\partial (\mathbf{m}_{k}^-)^2}=\mathbf{\Lambda}_k^- - \mathbf{H}_k^\top \mathbf{S}_k^{-1} \mathbf{H}_k
\end{align}
Combining both expressions recovers the MBF covariance-update equation:
\begin{align}    
    \mathbf{\Lambda}_k^-&=(\mathbf{I} - \mathbf{K}_k \mathbf{H}_k)^\top \mathbf{\Lambda}_k^+ (\mathbf{I} - \mathbf{K}_k \mathbf{H}_k) + \mathbf{H}_k^\top \mathbf{S}_k^{-1} \mathbf{H}_k
\end{align}
Together with the MBF covariance-prediction equation, the second-order adjoint recursions are therefore:
\begin{align}
    \mathbf{\Lambda}_k^+&=\mathbf{A}_k^\top \mathbf{\Lambda}_{k+1}^- \mathbf{A}_k \\
    \mathbf{\Lambda}_k^-&=(\mathbf{I} - \mathbf{K}_k \mathbf{H}_k)^\top \mathbf{\Lambda}_k^+ (\mathbf{I} - \mathbf{K}_k \mathbf{H}_k) + \mathbf{H}_k^\top \mathbf{S}_k^{-1} \mathbf{H}_k
\end{align}
which, by definition of $J_{k:T}$, is initialized from:
\begin{align}
    \mathbf{\Lambda}_T^+ = 0
\end{align}

\subsection{MBF Smoothing: Recovery of MAP Trajectory}
It remains to show how the smoothed means and covariances can be recovered from the adjoint variables. After eliminating the states up to time $k$, the remaining cost can be written in the form of \cref{eq:start_recovery} as:
\begin{align}
    \widetilde{L}_{k:T}(\mathbf{x}_k) = J_{1:k} + \frac{1}{2} \|\mathbf{x}_{k} - \mathbf{m}_{k} \|_{\mathbf{P}_{k}^{-1}}^2 + f_{k+1:T}(\mathbf{x}_k)
\end{align}
where $f_{k+1:T}(\mathbf{x}_k)$ collects all future cost terms expressed as a function of $\mathbf{x}_k$. Define the cost function as:
\begin{align}
    L_{k:T}(\mathbf{x}_k) = \frac{1}{2} \| \mathbf{x}_k - \mathbf{m}_{k} \|_{\mathbf{P}_{k}^{-1}}^2 + f_{k+1:T}(\mathbf{x}_k)
\end{align}
Then it follows that:
\begin{align}
    J_{k+1:T} =  \min_{\mathbf{x}_k} L_{k:T}(\mathbf{x}_k)
\end{align}
of which the minimizer is precisely the smoothed mean:
\begin{align}
    \mathbf{m}_{k}^s = \argmin_{\mathbf{x}_k} L_{k:T}(\mathbf{x}_k)
\end{align}
Consider now the derivative of the optimal cost $J_{k+1:T}$ with respect to the filtered mean $\mathbf{m}_k$:
\begin{align}
    \frac{\partial J_{k+1:T}}{\partial \mathbf{m}_{k}}&=\frac{\partial L_{k:T}}{\partial \mathbf{m}_{k}} \bigg\rvert_{\mathbf{x}_k = \mathbf{m}_{k}^s} + \frac{\partial L_{k:T}}{\partial \mathbf{x}_k}  \frac{\partial \mathbf{x}_k}{\partial \mathbf{m}_{k}} \bigg\rvert_{\mathbf{x}_k = \mathbf{m}_{k}^s} \\
    &=\mathbf{P}_{k}^{-1} (\mathbf{m}_{k} - \mathbf{m}_{k}^s)
\end{align}
since the second term vanishes due to the first order optimality condition:
\begin{align}\label{eq:FOP}
    \frac{\partial L_{k:T}}{\partial \mathbf{x}_{k}} \bigg\rvert_{\mathbf{x}_k = \mathbf{m}_{k}^s} &= \mathbf{P}_{k}^{-1} (\mathbf{m}_{k}^s - \mathbf{m}_{k}) + \frac{\partial f_{k+1:T}}{\partial \mathbf{x}_k}\bigg\rvert_{\mathbf{x}_k = \mathbf{m}_{k}^s} = 0.
\end{align}
Using the definition of $\bm{\lambda}_k^+$ from \cref{eq:lambda_plus_definition} directly results in its relation to the smoothed mean $\mathbf{m}_k^s$:
\begin{align} \label{eq:lambda_expr} 
\bm{\lambda}_k^+ 
    = 
    \mathbf{P}_k^{-1}(\mathbf{m}_k-\mathbf{m}_k^s) 
\end{align}

To recover the relation to the smoothed covariance, derive \cref{eq:lambda_expr} with respect to the filtered mean $\mathbf{m}_k$:
\begin{align}
    \bm{\Lambda}_k^+ = \left( \mathbf{I} - \frac{\partial \mathbf{m}_{k}^s}{\partial \mathbf{m}_{k}} \right)^\top \mathbf{P}_{k}^{-1} 
\end{align}
In order to obtain an expression for the partial derivative of the smoothed mean with respect to the filtered mean, derive the first-order optimality condition in \cref{eq:FOP} with respect to the filtered mean $\mathbf{m}_k$:
\begin{align}
    \mathbf{P}_{k}^{-1} \left(\frac{\partial \mathbf{m}_{k}^s}{\partial \mathbf{m}_{k}} - \mathbf{I} \right) + \frac{\partial^2 f_{k+1:T}}{\partial \mathbf{x}_k^2}\bigg\rvert_{\mathbf{x}_k= \mathbf{m}_{k}^s} \frac{\partial \mathbf{m}_{k}^s}{\partial \mathbf{m}_{k}} &= 0 \\
    \left( \mathbf{P}_{k}^{-1} + \frac{\partial^2 f_{k+1:T}}{\partial \mathbf{x}_k^2}\bigg\rvert_{\mathbf{x}_k = \mathbf{m}_{k}^s} \right)\frac{\partial \mathbf{m}_{k}^s}{\partial \mathbf{m}_{k}} &=\mathbf{P}_{k}^{-1}
\end{align}
The matrix between parentheses is the Hessian of the posterior cost at the optimum and is therefore equal to $[\mathbf{P}_k^s]^{-1}$, hence:
\begin{align}
    \frac{\partial \mathbf{m}_{k}^s}{\partial \mathbf{m}_{k}} = \mathbf{P}_k^s \mathbf{P}_k^{-1}
\end{align}
which directly relates the smoothed covariance to the adjoint variable $\bm{\Lambda}_k^+$:
\begin{align}
    \bm{\Lambda}_k^+ &= (\mathbf{I} - \mathbf{P}_k^s \mathbf{P}_k^{-1})^\top \mathbf{P}_k^{-1} \\
    &= \mathbf{P}_k^{-1} (\mathbf{P}_k - \mathbf{P}_k^s) \mathbf{P}_k^{-1}
\end{align}
A similar derivation can be done to determine the relation between the smoothed mean and covariance and the set $(\bm{\lambda}_k^-, \mathbf{\Lambda}_k^-)$.

\section{Derivation: Correction Equations}\label{app:correction_equations}
\citet{Parellier_Speeding_Up_Backprop} derived backward recursions for differentiating scalar objectives through the Kalman filter. We use these results to obtain the gradients of the negative log marginal likelihood with respect to the predicted and filtered covariance matrices, which are not directly provided by the MBF smoothing equations. Note that our definition of the negative log evidence in \cref{eq:JNLL_cost} differs by a factor of $2$ from the one presented in \citet{Parellier_Speeding_Up_Backprop}.

To simplify notation, define
\begin{align}
    \mathbf{B}_k = \mathbf{I} - \mathbf{K}_k \mathbf{H}_k
\end{align}
The corresponding covariance-gradient recursions are then given by \citep{Parellier_Speeding_Up_Backprop}:
\begin{align}
        \frac{\partial J_{1:T}^{\text{NLL}}}{\partial \mathbf{P}_{k}} &=  \mathbf{A}_k^\top  \frac{\partial J_{1:T}^{\text{NLL}}}{\partial \mathbf{P}_{k+1}^-}  \mathbf{A}_k     \\
        \begin{split}
        \frac{\partial J_{1:T}^{\text{NLL}}}{\partial \mathbf{P}_{k^-}}  &=  \mathbf{B}_k^\top \frac{\partial J_{1:T}^{\text{NLL}}}{\partial \mathbf{P}_{k}} \mathbf{B}_k + \frac{1}{2}     \mathbf{H}_k^\top \mathbf{S}_k^{-1}  \mathbf{H}_k \\
    &+ \frac{1}{2} \mathbf{B}_k^\top \bm{\lambda}_k^+  \mathbf{z}_k^\top  \mathbf{R}_k^{-1} \mathbf{H}_k \mathbf{B}_k \\
    &+ \frac{1}{2} \mathbf{B}_k^\top \mathbf{H}_k^\top  \mathbf{R}_k^{-1} \mathbf{z}_k \left( \bm{\lambda}_k^+ \right)^\top  \mathbf{B}_k \\
    &- \frac{1}{2} \mathbf{H}_k^\top  \mathbf{S}_k^{-1} \mathbf{z}_k \mathbf{z}_k^\top  \mathbf{S}_k^{-1} \mathbf{H}_k
    \end{split}
\end{align}
Substituting the decompositions in \crefrange{eq:def_T}{eq:decomposition} into these recursions separates the full covariance gradients into the MBF contribution and a remaining correction term. Using the MBF smoothing equations to cancel the terms involving ($\bm{\Lambda}_k^+, \bm{\Lambda}_k^-$) then yields the correction recurrences in \crefrange{eq:corr_rec_1}{eq:corr_rec_3} for ($\widetilde{\bm{\Lambda}}_k^+, \widetilde{\bm{\Lambda}}_k^-$).

\section{Derivation: Gradient Expressions}\label{app:gradient_expr}
This Appendix derives the individual terms appearing in \crefrange{eq:gradient_expression_conceptual}{eq:gradient_expr}. For compactness, we use index notation and assume summation over repeated indices. 

\subsection{Derivative with respect to $\mathbf{A}_k$}
By the chain rule,
\begin{align}
    \frac{\partial J_{1:T}^{\text{NLL}}}{\partial \left[ \mathbf{A}_k \right]_{ij}}  &= \frac{\partial J_{1:T}^{\text{NLL}}}{\partial \left[ \mathbf{m}_{k+1}^- \right]_p}\frac{\partial \left[ \mathbf{m}_{k+1}^- \right]_p}{\partial [\mathbf{A}_k]_{ij}} \\
    & + \ \frac{\partial J_{1:T}^{\text{NLL}}}{\partial \left[ \mathbf{P}_{k+1}^- \right]_{rs}} \frac{\partial \left[ \mathbf{P}_{k+1}^- \right]_{rs}}{\partial [\mathbf{A}_k]_{ij}}
\end{align}
The first term follows from:
\begin{align}
    \frac{\partial J_{1:T}^{\text{NLL}}}{\partial \left[ \mathbf{m}_{k+1}^- \right]_p} &= \frac{\partial J_{k+1:T}^{\text{NLL}}}{\partial \left[ \mathbf{m}_{k+1}^- \right]_p} \\
    &= [\bm{\lambda}_{k+1}^-]_p \\
    \frac{\partial \left[ \mathbf{m}_{k+1}^- \right]_p}{\partial [\mathbf{A}_k]_{ij}} &= \frac{\partial \left[ \mathbf{A}_k \right]_{pq} \left[ \mathbf{m}_{k} \right]_q}{\partial [\mathbf{A}_k]_{ij}} \\
    &= \delta_{ip} \delta_{jq} \left[ \mathbf{m}_{k} \right]_q \\
    &= \delta_{ip} \left[ \mathbf{m}_{k} \right]_j
\end{align}
Hence,
\begin{align}
    \frac{\partial J_{1:T}^{\text{NLL}}}{\partial \left[ \mathbf{m}_{k+1}^- \right]_p}   \frac{\partial \left[ \mathbf{m}_{k+1}^- \right]_p}{\partial [\mathbf{A}_k]_{ij}} &= [\bm{\lambda}_{k+1}^-]_p \delta_{ip} \left[ \mathbf{m}_{k} \right]_j \\
    &= [\bm{\lambda}_{k+1}^- \mathbf{m}_k^\top]_{ij}
\end{align}
Considering the second term:
\begin{align}
    \frac{\partial J_{1:T}^{\text{NLL}}}{\partial \left[ \mathbf{P}_{k+1}^- \right]_{rs}} &= \frac{\partial J_{k+1:T}^{\text{NLL}}}{\partial \left[ \mathbf{P}_{k+1}^- \right]_{rs}} \\
    &= [\mathbf{T}_{k+1}^-]_{rs} \\
     \frac{\partial \left[ \mathbf{P}_{k+1}^- \right]_{rs}}{\partial [\mathbf{A}_k]_{ij}} &= \frac{\partial ([\mathbf{A}_k]_{rt} [\mathbf{P}_k]_{tu} [\mathbf{A}_k]_{su} + [\mathbf{Q}_k]_{rs})}{\partial [\mathbf{A}_k]_{ij}} \\
     &= \delta_{ri} [\mathbf{P}_k \mathbf{A}_k^\top]_{js} + \delta_{si} [\mathbf{A}_k \mathbf{P}_k]_{rj}
\end{align}
Hence, this simplifies to:
\begin{align}
     \frac{\partial J_{1:T}^{\text{NLL}}}{\partial \left[ \mathbf{P}_{k+1}^- \right]_{rs}} \frac{\partial \left[ \mathbf{P}_{k+1}^- \right]_{rs}}{\partial [\mathbf{A}_k]_{ij}} &= [\mathbf{T}_{k+1}^-]_{rs}  \delta_{ri}  [\mathbf{P}_k \mathbf{A}_k^\top]_{js} \\
     &+ \, [\mathbf{T}_{k+1}^-]_{rs} \delta_{si} [\mathbf{A}_k \mathbf{P}_k]_{rj} \\
     &= 2 [\mathbf{T}_{k+1}^- \mathbf{A}_k \mathbf{P}_k]_{ij}
\end{align}
Therefore,
\begin{align}
    \frac{\partial J_{1:T}^{\text{NLL}}}{\partial\mathbf{A}_k} = \bm{\lambda}_{k+1}^- \mathbf{m}_k^\top +2 \mathbf{T}_{k+1}^- \mathbf{A}_k \mathbf{P}_k
\end{align}
\subsection{Derivative with respect to $\mathbf{Q}_k$}
Again by the chain rule,
\begin{align}
    \frac{\partial J_{1:T}^{\text{NLL}}}{\partial \left[ \mathbf{Q}_k \right]_{ij}}  &= \frac{\partial J_{1:T}^{\text{NLL}}}{\partial \left[ \mathbf{m}_{k+1}^- \right]_p}\frac{\partial \left[ \mathbf{m}_{k+1}^- \right]_p}{\partial [\mathbf{Q}_k]_{ij}} \\
    & + \ \frac{\partial J_{1:T}^{\text{NLL}}}{\partial \left[ \mathbf{P}_{k+1}^- \right]_{rs}} \frac{\partial \left[ \mathbf{P}_{k+1}^- \right]_{rs}}{\partial [\mathbf{Q}_k]_{ij}}
\end{align}
The predicted mean $\mathbf{m}_{k+1}^-$ does not depend on the process noise covariance matrix $\mathbf{Q}_k$, so the first term vanishes. For the second term:
\begin{align}
    \frac{\partial J_{1:T}^{\text{NLL}}}{\partial \left[ \mathbf{P}_{k+1}^- \right]_{rs}} &= [\mathbf{T}_{k+1}^-]_{rs} \\ \frac{\partial \left[ \mathbf{P}_{k+1}^- \right]_{rs}}{\partial [\mathbf{Q}_k]_{ij}} &= \frac{\partial ([\mathbf{A}_k]_{rt} [\mathbf{P}_k]_{tu} [\mathbf{A}_k]_{su} + [\mathbf{Q}_k]_{rs})}{\partial [\mathbf{Q}_k]_{ij}} \\
    &= \delta_{ri} \delta_{sj}
\end{align}
Hence this simplifies to:
\begin{align}
    \frac{\partial J_{1:T}^{\text{NLL}}}{\partial \left[ \mathbf{P}_{k+1}^- \right]_{rs}} \frac{\partial \left[ \mathbf{P}_{k+1}^- \right]_{rs}}{\partial [\mathbf{Q}_k]_{ij}} = [\mathbf{T}_{k+1}^-]_{ij}
\end{align}
Therefore,
\begin{align} \frac{\partial J_{1:T}^{\mathrm{NLL}}} {\partial \mathbf{Q}_k} = \mathbf{T}_{k+1}^-
\end{align}

\subsection{Derivative with respect to $\mathbf{P}_0$}
Finally, by definition of $\mathbf{T}_0^+$,
\begin{align}
    \frac{\partial J_{1:T}^{\text{NLL}}}{\partial \mathbf{P}_{0}} = \mathbf{T}_{0}^+
\end{align}

\subsection{Derivative with respect to $\mathbf{R}_k$}
In this work, we do not consider the derivative of the negative log marginal likelihood w.r.t. the measurement noise covariance $\mathbf{R}_k$. For an expression, we refer to Eq.~(26) in Section III.C of \citep{Parellier_Speeding_Up_Backprop}.

\section{Numerical Analysis}\label{app:flop_count}
This Appendix summarizes the computational and memory requirements of the filtering, smoothing, and hyperparameter learning algorithms considered in this work. \Crefrange{tab:kalman_predictor}{tab:hyperparameter_tuning} report the corresponding floating-point operation (flop) counts, where a flop represents a single operation (e.g. addition, multiplication). Symmetry is exploited whenever applicable, and we assume $m < d$ when choosing operation ordering.

The smoothing methods also differ in their memory requirements. The following analysis assumes that symmetric matrices are stored using only their upper or lower triangular entries. The standard RTS smoother stores the predicted and filtered means and covariances at every step, as well as the transition matrices, resulting in a total memory requirement of:
\begin{align}
    T(2d^2 + 3d)
\end{align}
expressed in the number of stored scalar entries. The Joseph-form RTS smoother additionally requires the process noise covariance matrices, Kalman gains, and observation matrices, leading to:
\begin{align}
    T \frac{1}{2} (5d^2 + 7d + 4md)
\end{align}
The MBF smoother stores the transition matrices, Kalman gains, observation matrices, Cholesky factors of the innovation covariance, and measurement residuals. If the smoothed solution is reconstructed only at $T_\mathrm{sol}$ out of a total of $T$ time steps, the total memory requirement is:
\begin{align}
    T(d^2 + 2dm +\frac{1}{2}m^2 + \frac{3}{2}m) + T_{\text{sol}} \frac{1}{2} (d^2 + 3d)
\end{align}

\begin{table*}[t!]\centering
\caption{Flop counts of the Kalman predictor equations.}
\begin{tabular*}{\textwidth}{@{\extracolsep{\fill}} L{0.7\textwidth} R{0.22\textwidth}}
\toprule
Equation & \multicolumn{1}{r}{Flop Count} \\
\midrule
$\mathbf{m}_k^- = \mathbf{A}_{k-1} \mathbf{m}_{k-1}$
& 2d^2 -d \\
$\mathbf{P}_{k}^- = \mathbf{A}_{k-1} \mathbf{P}_{k-1} \mathbf{A}_{k-1}^\top + \mathbf{Q}_{k-1}$
& 3d^3\\
\bottomrule
\end{tabular*}
\label{tab:kalman_predictor}
\end{table*}

\begin{table*}[t!]\centering
\caption{Flop counts of the Kalman filtering equations.}
\begin{tabular*}{\textwidth}{@{\extracolsep{\fill}} L{0.35\textwidth} R{0.4\textwidth}}
\toprule
Equation &  \multicolumn{1}{r}{Flop Count} \\
\midrule
$\mathbf{z}_k = \mathbf{y}_k - \mathbf{H}_k \mathbf{m}_{k}^-$
& 2dm \\
$\mathbf{S}_k = \mathbf{H}_k \mathbf{P}_{k}^- \mathbf{H}_k^\top + \mathbf{R}_k$
& 2d^2m + dm^2 \\
$\mathbf{K}_k = \mathbf{P}_{k}^- \mathbf{H}_k^\top \mathbf{S}_k^{-1}$
& 2 d^2m + d(2m^2 - m) + \frac{1}{3}m^3 + \frac{1}{2}m^2 + \frac{1}{6}m \\
$\mathbf{m}_{k} = \mathbf{m}_{k}^- + \mathbf{K}_k \mathbf{z}_k$
& 2dm  \\
$\mathbf{P}_{k} = (\mathbf{I} - \mathbf{K}_k \mathbf{H}_k) \mathbf{P}_{k}^-$
& 3 d^2 m \\
Joseph form
& 7d^2m + 2dm^2 \\
\bottomrule
\end{tabular*}
\label{tab:kalman_filter}
\end{table*}

\begin{table*}[t!]\centering
\caption{Flop counts of the RTS smoothing equations.}
\begin{tabular*}{\textwidth}{@{\extracolsep{\fill}} L{0.6\textwidth} R{0.35\textwidth}}
\toprule
Equation &  \multicolumn{1}{r}{Flop Count} \\
\midrule
$\mathbf{J}_k = \mathbf{P}_{k} \mathbf{A}_k^\top [\mathbf{P}_{k+1}^-]^{-1}$
& \frac{13}{3} d^3 - \frac{1}{2}d^2 + \frac{1}{6}d \\
$\mathbf{m}_{k}^s = \mathbf{m}_{k} + \mathbf{J}_k (\mathbf{m}_{k+1}^s - \mathbf{m}_{k+1}^-)$
& 2 d^2 + d \\
$\mathbf{P}_{k}^s =
\mathbf{P}_{k} + \mathbf{J}_k (\mathbf{P}_{k+1}^s - \mathbf{P}_{k+1}^-) \mathbf{J}_k^\top$
& 3 d^3 + \frac{1}{2} d^2 + \frac{1}{2} d \\
Joseph form
& 7d^3 + \frac{5}{2}d^2 + \frac{3}{2}d \\
\bottomrule
\end{tabular*}
\label{tab:rts_smoother}
\end{table*}

\begin{table*}[t!]\centering
\caption{Flop counts of the MBF smoothing equations.}
\begin{tabular*}{\textwidth}{@{\extracolsep{\fill}} L{0.7\textwidth} R{0.25\textwidth}}
\toprule
Equation & \multicolumn{1}{r}{Flop Count} \\
\midrule
$\bm{\lambda}_{k-1}^+ = \mathbf{A}_{k-1}^\top \bm{\lambda}_{k}^-$
& 2 d^2 - d \\
$\mathbf{\Lambda}_{k-1}^+ = \mathbf{A}_{k-1}^\top \mathbf{\Lambda}_k^- \mathbf{A}_{k-1}$
& 3 d^3 - \frac{1}{2} d^2 - \frac{1}{2} d \\
$\bm{\lambda}_k^- =
(\mathbf{I} - \mathbf{K}_k \mathbf{H}_k)^\top \bm{\lambda}_k^+
- \mathbf{H}_k^\top \mathbf{S}_k^{-1} \mathbf{z}_k$
& 6 d m + 2 m^2 - m \\
$\mathbf{\Lambda}_k^- =
(\mathbf{I} - \mathbf{K}_k \mathbf{H}_k)^\top
\mathbf{\Lambda}_k^+
(\mathbf{I} - \mathbf{K}_k \mathbf{H}_k)
+ \mathbf{H}_k^\top \mathbf{S}_k^{-1} \mathbf{H}_k$
& d^2(7m + \frac{1}{2}) + d(4m^2 + 2m + \frac{1}{2}) \\
$\mathbf{m}_{k}^s = \mathbf{m}_{k} - \mathbf{P}_{k} \bm{\lambda}_k^+$
& 2 d^2 \\
$\mathbf{P}_{k}^s = \mathbf{P}_{k} - \mathbf{P}_{k} \mathbf{\Lambda}_k^+ \mathbf{P}_{k}$
& 3d^3 \\
\bottomrule
\end{tabular*}
\label{tab:mbf_smoother}
\end{table*}

\begin{table*}[t!]\centering
\caption{Flop counts of the MBF correction equations.}
\begin{tabular*}{\textwidth}{@{\extracolsep{\fill}} L{0.6\textwidth} R{0.35\textwidth}}
\toprule
Equation & \multicolumn{1}{r}{Flop Count} \\
\midrule
$\mathbf{l}_k = \frac{1}{2} \bm{\lambda}_k^+ \mathbf{z}_k^\top \mathbf{R}_k^{-1} \mathbf{H}_k$ & d(2m^2 + 3m -1) + \frac{1}{3}m^3 + \frac{1}{2}m^2 + \frac{1}{6}m \\
$\widetilde{\mathbf{\Lambda}}_k^- = (\mathbf{I} - \mathbf{K}_k \mathbf{H}_k)^\top (\widetilde{\mathbf{\Lambda}}_k^+ + \mathbf{l}_k + \mathbf{l}_k^\top) (\mathbf{I} - \mathbf{K}_k \mathbf{H}_k) -\frac{1}{2} \mathbf{H}_k^\top \mathbf{S}_k^{-1} \mathbf{H}_k$ \\
$\qquad \ \ \, - \frac{1}{2}\mathbf{H}_k^\top \mathbf{S}_k^{-1} \mathbf{z}_k \mathbf{z}_k^\top \mathbf{S}_k^{-1} \mathbf{H}_k$ & d^2 (7m + 4) + d(2m^2 + 3m + 2) + 2m^2\\ 
 $\widetilde{\mathbf{\Lambda}}_{k-1}^+ = \mathbf{A}_{k-1}^\top \widetilde{\mathbf{\Lambda}}_k^- \mathbf{A}_{k-1}$ & 3 d^3 - \frac{1}{2} d^2 - \frac{1}{2} d \\
\bottomrule
\end{tabular*}
\label{tab:MBF_correction_equations}
\end{table*}

\begin{table*}[t!]\centering
\caption{Flop counts of the gradient components of the negative log marginal likelihood.}
\begin{tabular*}{\textwidth}{@{\extracolsep{\fill}} L{0.7\textwidth} R{0.25\textwidth}}
\toprule
Equation & \multicolumn{1}{r}{Flop Count} \\
\midrule
$\sum_{k=0}^{T-1} \left\langle ( \bm{\lambda}_{k+1}^- \mathbf{m}_{k}^\top + 2\mathbf{T}_{k+1}^- \mathbf{A}_k \mathbf{P}_{k}), \frac{\text{d}\mathbf{A}_k}{\text{d}\theta} \right\rangle$ & T(3d^3 + \frac{5}{2}d^2 + \frac{1}{2}d) \\
$\sum_{k=0}^{T-1} \, \langle \mathbf{T}_{k+1}^-, \frac{\text{d}\mathbf{Q}_k}{\text{d}\theta} \rangle$ & T(d^2 + d) \\
$\left \langle \mathbf{T}_0^+,\frac{\text{d}\mathbf{P}_0}{\text{d}\theta} \right\rangle$ & d^2 + d \\
\bottomrule
\end{tabular*}
\label{tab:hyperparameter_tuning}
\end{table*}

\end{document}